\documentclass[aap,MSNbibl,dvips]{arximspdf}

\doi{10.1214/13-AAP940} %kopijuoti is PTS
\volume{24}
\issue{3}
\pubyear{2014}
\firstpage{1002}
\lastpage{1048}

\makeatletter
\newcommand{\rrvert}{\vert}
\newcommand{\llvert}{\vert}
\newtheorem{lemma}{Lemma}
\newtheorem{corollary}{Corollary}
\newtheorem{theorem}{Theorem}
\newtheorem{proposition}{Proposition}
\newproclaim{remark}{Remark}
\newproclaim{example}{Example}
\newproclaim{definition}{Definition}
\newcommand{\eqref}[1]{(\ref{#1})}
\makeatother

\begin{document}
\begin{frontmatter}

\title{Asymptotically optimal discretization of hedging strategies with jumps}
\runtitle{Asymptotically optimal discretization}

\begin{aug}
\author[A]{\fnms{Mathieu} \snm{Rosenbaum}\corref{}\ead[label=e1]{mathieu.rosenbaum@polytechnique.edu }}
\and
\author[B]{\fnms{Peter} \snm{Tankov}}
\runauthor{M. Rosenbaum and P. Tankov}
\affiliation{LPMA, Universit\'e Pierre et Marie Curie and LPMA,\break
Universit\'e Paris Diderot}
\address[A]{LPMA\\
Universit\'e Pierre et Marie Curie\\
(Paris 6)\\
Case courrier 188, 4 place Jussieu\\
75252 Paris Cedex 05\\
France} %adresu isvedimo komanda gale!
\address[B]{LPMA\\
Universit\'e Paris Diderot\\
(Paris 7)\\
Case courrier 7012\\
75205 Paris Cedex 13\\
France}
\end{aug}

% HISTORY:
\received{\smonth{8} \syear{2011}}
\revised{\smonth{5} \syear{2013}}

% ABSTRACT
%
\begin{abstract}
In this work, we consider the hedging error due to
discrete trading in models with jumps. Extending an approach
developed by Fukasawa [In \textit{Stochastic Analysis with Financial Applications} (2011) 331--346
Birkh\"auser/Springer Basel AG] for continuous processes, we propose
a framework enabling us to (asymptotically) optimize the
discretization times. More precisely, a discretization rule is
said to be optimal if for a given cost function, no strategy has
(asymptotically, for large cost) a lower mean square
discretization error for a smaller cost. We focus on
discretization rules based on hitting times and give explicit
expressions for the optimal rules within this class.
\end{abstract}

% KEYWORDS

\end{frontmatter}

%s1 #&#
\section{Introduction}\label{sec1}
A basic problem in mathematical finance is how to replicate a random
claim with $\mathcal F_T$-measurable payoff $H_T$ with a portfolio
involving only the underlying asset $Y$ and cash.
When $Y$ follows a diffusion process of the form
%
%e1 #&#
%
\begin{equation}
dY_t = \mu(t,Y_t)\,dt + \sigma(t,Y_t)
\,dW_t,\label{diff}
\end{equation}
it is known that under minimal assumptions, a random payoff depending
only on the terminal value of the asset $H_T = H(Y_T)$ can be
replicated with the so-called delta hedging strategy. This means that the
number of units of the underlying asset to hold at time $t$ is equal to $X_t
= \frac{\partial P(t,Y_t)}{\partial Y}$,
where $P(t,Y_t)$ is the price of the option, which is uniquely defined
in such a model. However, to implement such a strategy, the hedging
portfolio must be readjusted continuously, which is of course
physically impossible and irrelevant because of the presence of
microstructure effects and transaction costs. For this reason,
the optimal strategy is always replaced with a piecewise constant one,
leading to a discretization
error. The relevant questions are then: (i) how big is this
discretization error, and (ii) when are the good times to readjust the
hedge?

Assume first that the hedging portfolio is readjusted at
regular intervals of length $h = \frac{T}{n}$. A result by
Zhang \cite{zhang.couverture} (see also
\cite{bertsimas.kogan.lo.00,hayashi.mykland.05}) then shows that
for Lipschitz continuous payoff functions, assuming zero interest
rates, the discretization error
\[
\mathcal E_T^n = \int_0^T
X_t \,dY_t - \int_0^T
X_{h[t/h]} \,d Y_t
\]
satisfies
%
%e2 #&#
%
\begin{equation}
\lim_{h\to0} n E\bigl[\bigl(\mathcal E^n_T
\bigr)^2\bigr] = \frac{T}{2}E \biggl[\int_0^T
\biggl(\frac{\partial^2 P}{\partial Y^2} \biggr)^2 \sigma(s,Y_s)^4
\,ds \biggr].\label{zhang}
\end{equation}
Of course, it is intuitively clear that readjusting the portfolio at
regular deterministic
intervals is not optimal. However, the optimal strategy for fixed $n$
is very difficult to compute.

Fukasawa \cite{fukasawa.09b} simplifies this problem by
assuming that the hedging portfolio is readjusted at high
frequency. The performance of different families of strategies
can then be compared based on their asymptotic behavior as the
number of readjustment dates $n$ tends to infinity, rather than
the performance for fixed $n$. Consider a
sequence of discretization strategies
\[
0 = T^n_0 < T^n_1 < \cdots<
T^n_j < \cdots,
\]
with $\sup_{j} |T^n_{j+1}-T^n_j|\to0$ as $n\to\infty$, and
let $N^n_T:= \max\{j\geq0; T^n_j \leq T\}$ be the total
number of readjustment dates on the interval $[0,T]$ for given $n$.
%To compare discretization rules, we need a functional which
%on the asymptotic behavior and
%discretization error and the cost of readjusting the portfolio.
To compare two such sequences in terms of their
asymptotic behavior for large $n$, Fukasawa~\cite{fukasawa.09b} uses the
functional
%
%e3 #&#
%
\begin{equation}
\lim_{n\to\infty} E\bigl[N^n_T\bigr] E\bigl[
\bigl\langle\mathcal E^n\bigr\rangle_T\bigr],\label{fuklim}
\end{equation}
where $\langle\mathcal E^n\rangle$ is the quadratic variation of
the semimartingale $(\mathcal E^n_t)_{t\geq0}$. He finds that
when the underlying asset is a continuous semimartingale, the
functional \eqref{fuklim} admits a nonzero lower bound over all
such sequences, and exhibits a specific sequence which attains
this lower bound and is therefore called \emph{asymptotically
efficient}.

In the diffusion model \eqref{diff}, the asymptotically
efficient sequence takes the form
%
%e4 #&#
%
\begin{eqnarray}\label{fukrule}
T^n_{j+1} &=& \inf\biggl\{t>T^n_j;
|X_t - X_{T^n_j}|^2 \geq h_n
\frac{\partial^2 P(T^n_j,Y_{T^n_j})}{\partial
Y^2}\biggr\},
\nonumber
\\[-8pt]
\\[-8pt]
\nonumber
 X_t &=& \frac{\partial P(t,Y_t)}{\partial
Y},
\end{eqnarray}
where $h_n$ is a deterministic sequence with $h_n \to0$. In this case,
%
%e5 #&#
%
\begin{equation}
\lim_{n\to\infty} E\bigl[N^n_T\bigr] E\bigl[
\bigl\langle\mathcal E^n\bigr\rangle_T\bigr] =
\frac{1}{6}E \biggl[\int_0^T
\frac{\partial^2 P}{\partial Y^2} \sigma (s,Y_s)^2 \,ds
\biggr]^2,\label{diffeff}
\end{equation}
whereas for readjustment at equally spaced dates, formula
\eqref{zhang} yields
%
%e6 #&#
%
\begin{equation}
\lim_{n\to\infty} E\bigl[N^n_T\bigr] E\bigl[
\bigl\langle\mathcal E^n\bigr\rangle_T\bigr] =
\frac{T}{2}E \biggl[\int_0^T \biggl(
\frac{\partial^2 P}{\partial Y^2} \biggr)^2 \sigma(s,Y_s)^4
\,ds \biggr].\label{diffunif}
\end{equation}
Using the Cauchy--Schwarz inequality, we then see that the
asymptotically efficient discretization leads to a gain of at least a
factor $3$, compared to readjustment
at regularly spaced points.

Remark that the discretization scheme \eqref{fukrule} is
very different from the classical approximation schemes for stochastic
differential equations such as Euler or Milstein schemes. In order to
be implemented it requires the continuous observation of $(X_t)$ and
$(Y_t)$, which of course makes sense in the mathematical finance
context because the prices are, essentially, continuously observable
and the need
for discretization is due to the presence of transaction costs.

While the above approach is quite natural and provides
very explicit results, it fails to take into account important
factors of market reality. First, the asymptotic functional
\eqref{fuklim} is somewhat ad hoc, and does not reflect any
specific model for the transaction costs. Yet, transaction costs
are one of the main reasons why continuous (or almost continuous)
readjustments are not used. Therefore, they should be the
determining factor for any discretization algorithm. On the other
hand, the continuity assumption, especially at relatively high
frequencies, is not realistic. Indeed, it is well known that jumps
in the price occur quite frequently and have a significant impact
on the hedging error. It can even be argued that high-frequency
financial data are best described by pure jump processes; see
\cite{finestructure}.

The objective of this paper is therefore two-fold.
First, we develop a framework for characterizing the asymptotic
efficiency of discretization strategies which takes into account
the transaction costs. Second, we remove the continuity assumption
in order to understand the effect of the activity of small jumps
(often quantified by the Blumenthal--Getoor index) on the optimal
discretization strategies.

Models with jumps correspond to incomplete markets,
where the hedging issue is an approximation problem,
%
%e7 #&#
%
\begin{equation}
\min_X E \biggl(c + \int_0^T
X_{t-} \,dY_t - H_T \biggr)^2,\label{quadcrit}
\end{equation}
where $Y$ is now a semimartingale with jumps. The optimal strategy
$X^*$ for this problem is known to exist for any $H_T \in L^2$; see
\cite{cont.al.05,follmer.sondermann.86,follmer.schweizer.91,schweizer.01,kallsen.hubalek.al.06,cerny.kallsen.07}.
If the expectation in \eqref{quadcrit} is computed under a
martingale probability measure, then for any admissible strategy
$X'$,
%
%e8 #&#
%
\begin{eqnarray}\label{pythagore}
E \biggl(c + \int_0^T X'_{t-}
\,dY_t - H_T \biggr)^2 &=& E \biggl(\int
_0^T \bigl(X'_{t-}-X^*_{t-}
\bigr) \,dY_t \biggr)^2
\nonumber
\\[-8pt]
\\[-8pt]
\nonumber
&&{}+ E \biggl(c + \int
_0^T X^*_{t-} \,dY_t -
H_T \biggr)^2.
\end{eqnarray}
Indeed, $\int X_{t-}^* \,dY_t$ is essentially the orthogonal projection of
$H_T$ on the subspace of $L^2$ constituted by the stochastic integrals
of the form $\int X_{t-} \,dY_t$ where $X_{t-}$ is an admissible hedging
strategy.
Therefore, the quadratic hedging problem \eqref{quadcrit} and the
discretization problem can be studied separately. Given that the
quadratic hedging problem has already been studied by many authors,
in this paper we concentrate on the discretization problem.

Our goal is to study and compare discretization {rules}
for stochastic integrals of the form
\[
\int_0^TX_{t-} \,dY_t,
\]
where $X_t$ and $Y_t$ are semimartingales
with jumps, with the aim of identifying asymptotically optimal
{rules}. In particular we wish to understand the impact of the small
jumps of $X$ on the discretization error, and therefore we assume
that $X$ has no continuous local martingale part; see Remark \ref{purejumprem}.

A \emph{discretization rule} is a family of stopping
times $(T_i^\varepsilon)_{i\geq0}^{\varepsilon>0}$ parameterized
by a nonnegative integer $i$ and a positive real $\varepsilon$,
such that for every $\varepsilon>0$, $0 = T^\varepsilon_0 <
T^\varepsilon_1 < T^\varepsilon_2 < \cdots$. For a fixed discretization rule
and a fixed $\varepsilon$, we let $\eta^\varepsilon(t) =
\sup\{T^\varepsilon_i\dvtx T^\varepsilon_i \leq t\}$ and $N_T^{\varepsilon} =
\sup\{i\dvtx T^\varepsilon_i \leq T\}$. Motivated by decomposition~\eqref{pythagore}, we measure the performance of a
discretization rule with the $L^2$ error functional
%
%e9 #&#
%
\begin{equation}
\mathcal E(\varepsilon):= E \biggl[ \biggl(\int_0^T
(X_{t-} - X_{\eta
(t)-}) \,dY_t \biggr)^2
\biggr].\label{errfunc}
\end{equation}
Also, to each discretization rule we associate a family of cost
functionals of the form
%
%e10 #&#
%
\begin{equation}
\mathcal C^\beta(\varepsilon) = E \biggl[\sum
_{i\geq1: T^\varepsilon_i
\leq T} |X_{T^\varepsilon_i} - X_{T^\varepsilon_{i-1}}|^\beta
\biggr],\label{costfunc}
\end{equation}
with $\beta\in[0,2]$.
The case $\beta=0$ corresponds to a fixed cost per transaction, and
the case $\beta=1$ corresponds to a fixed cost per unit of asset.
{Other values of $\beta$ often appear in the market
microstructure literature where one considers that transaction costs
are explained by the shape of the order book.}

In our
framework, a discretization rule is said to be optimal
for a given cost functional if no strategy has (asymptotically,
for large costs) a lower discretization error and a smaller
cost.

Motivated by the representation
\eqref{fukrule} and the readjustment rules used by market
practitioners, we focus on discretization strategies based on the
exit times of $X$ out of random intervals
%
%e11 #&#
%
\begin{equation}
T^\varepsilon_{i+1} = \inf\bigl\{t>T^\varepsilon_i
\dvtx X_t\notin (X_{T^\varepsilon_i}- \varepsilon\underline{a}_{T^\varepsilon_i},
X_{T^\varepsilon
_i}+\varepsilon\overline{a}_{T^\varepsilon_i}) \bigr\},\label{discrule}
\end{equation}
where $(\overline a_t)_{t\geq0}$ and $(\underline a_t)_{t\geq
0}$ are positive $\mathbb F$-adapted c\`adl\`ag
processes.

In Theorems \ref{err.thm} and \ref{cost.thm}, we characterize
explicitly the asymptotic behavior of the errors and
costs associated to these random discretization rules, by showing that,
under suitable assumptions,
\begin{eqnarray*}
\lim_{\varepsilon\to0} \varepsilon^{-2} \mathcal E(\varepsilon)
&=& E \biggl[\int_0^T {A_t}
\frac{f(\underline a_t, \overline
a_t)}{g(\underline a_t, \overline a_t)} \,dt \biggr],
\\
\lim_{\varepsilon\to0} \varepsilon^{\alpha-\beta} \mathcal
C^{ \beta
}(\varepsilon) &=& E \biggl[\int_0^T
{\lambda_t} \frac{u^\beta(\underline a_t, \overline
a_t)}{g(\underline a_t, \overline a_t)} \,dt \biggr],
\end{eqnarray*}
where, for $\underline a, \overline a \in(0,\infty)$,
\begin{eqnarray*}
f(\underline a, \overline a) &=& E \biggl[\int_0^{\tau^*}
\bigl(X^*_t\bigr)^2 \,dt \biggr],\qquad g(\underline a, \overline
a) = E\bigl[\tau^*\bigr]\quad \mbox {and}\\
u^\beta(\underline a, \overline a)
&=& E\bigl[\bigl|X^*_{\tau^*}\bigr|^\beta\bigr]<\infty,
\end{eqnarray*}
with $\tau^* = \inf\{t\geq0\dvtx X^*_t \notin(-\underline a, \overline a)
\}$, where $X^*$ is a strictly $\alpha$-stable process determined from
$X$ by a limiting procedure, and the processes $A$ and $\lambda$ are
determined from the semimartingale characteristics of $X$ and $Y$.

This allows us to determine the asymptotically optimal
intervals as
solutions to a simple optimization problem (Proposition \ref{optbar.prop}). In
particular, we show that in the case where the cost functional is
given by the expected number of discretization dates, the error
associated to our optimal strategy with the cost equal to $N$,
converges to zero as $N\to\infty$ at a faster rate than the error
obtained by readjusting at $N$ equally spaced dates.

{As applications of our method, we consider the
discretization of the hedging strategy for a European option in an
exponential L\'evy model (Proposition \ref{levy.prop}) and the
discretization of the Merton
portfolio strategy (Proposition~\ref{merton.prop}). In the option
hedging problem,} we obtain an
explicit representation for the optimal discretization dates,
which is similar to \eqref{fukrule}, but includes two ``tuning''
parameters: an index which determines the effect of transaction
costs (fixed, proportional, etc.) and the Blumenthal--Getoor index
measuring the activity of small jumps.

This paper is structured as follows. In Section~\ref{prelim}, we introduce our framework and in particular the notion
of asymptotic optimality based on the limiting
behavior of the error and cost functionals. The assumptions on the
processes $X$
and $Y$ and on the admissible discretization rules are also stated here.
Section~\ref{main.thm} contains the main results of this paper
which characterize the limiting behavior of the error and the cost
functionals, {and Sections~\ref{barrier.sec} to \ref{explevy.sec}
provide explicit examples of optimal
discretization strategies in various contexts.} Sections~\ref{proofthm} and
\ref{proofthm2} contain the proofs of the main results and Section~\ref{prooflm} gathers some technical lemmas needed in Section~\ref{proofthm2}.

%s2 #&#
\section{Framework}\label{prelim}
%We consider a probability sace $(\Omega,\mathcal F, \mathbb P)$ and a
%filtration $\mathbb F = (\mathcal F_t)_{t\geq0}$, to which are
%adapted a one-dimensional Brownian motion $W$ and a Poisson random
%measure $J$ on $[0,\infty)\times\mathbb R$ with intensity $dt \times
%$\int_\mathbb R(1\wedge x^2)\nu(dx)<\infty$. $\tilde J$ will denote the
%compensated version of $J$.

%For every pair of processes $(\overline a,\underline a)$, there
%is a family of discretization strategies parameterized by
%$\varepsilon>0$.

%pa2.subsection.subsubsection.1 #&#
\textit{Asymptotic comparison of discretization rules}.
We are \mbox{interested} in comparing different discretization rules, as
defined in the \hyperref[sec1]{Intro-}\break \hyperref[sec1]{duction}, for the
stochastic integral
\[
\int_0^T X_{t-} \,dY_t,
\]
where $X$ and $Y$ are semimartingales, in terms of their limiting
behavior when the number of discretization points tends to
infinity.

The performance of a given discretization rule is assessed by
the error functional $\mathcal E(\varepsilon) \dvtx (0,\infty)\to
[0,\infty)$ {(which measures the discretization error associated
to this rule)} and a cost functional $\mathcal C^{\beta}(\varepsilon) \dvtx (0,\infty)\to[0,\infty)$ {(which measures the corresponding
transaction cost)}, as defined in \eqref{errfunc} and \eqref{costfunc}.
% $\mathcal C^\varepsilon_T$: $(\overline a,\underline a, \varepsilon)
% \mapsto[0,\infty)$.
We assume that the cost functional is such that
\[
\lim_{\varepsilon\downarrow0} \mathcal C^{\beta}(\varepsilon) = + \infty.
\]
For $C>0$ sufficiently large, we define
\[
\varepsilon(C) = \inf\bigl\{\varepsilon>0 \dvtx \mathcal C^\beta(
\varepsilon)<C\bigr\}
\]
and $\overline{\mathcal E}(C):=\mathcal E(\varepsilon(C))$.
%
%de1 #&#
\begin{definition}\label{assdom}
We say
that the discretization rule $A$ asymptotically dominates the rule $B$ if
\[
\limsup_{C\to\infty} \frac{\overline{\mathcal E}^A(C)}{\overline
{\mathcal E}^B(C)}\leq 1.
\]
\end{definition}
To apply Definition \ref{assdom}, the following
{simple result} will be very useful.
%
%le1 #&#
\begin{lemma}\label{errorcost.lm}
Assume that for a given discretization rule, the cost and error
functionals are such that there exist
$a>0$ and $b>0$ with
%
%e12 #&#
%
\begin{equation}
\lim_{\varepsilon\downarrow0} \varepsilon^{-a}\mathcal E(\varepsilon)
= \hat{\mathcal E} \quad\mbox{and}\quad \lim_{\varepsilon\downarrow0} \varepsilon^{b}
\mathcal C^{\beta
}(\varepsilon) = \hat{\mathcal C} \label{errorcost}
\end{equation}
for some {positive} constants $\hat{\mathcal E}$ and $\hat{\mathcal
C}$. Then
\[
\overline{ \mathcal E}(C) \sim C^{-{a}/{b}} (\hat{\mathcal
C})^{{a}/{b}} \hat{\mathcal E} \qquad\mbox{as } C\to\infty.
\]
\end{lemma}
We shall consider discretizations based on the
hitting times of the process~$X$. Recall that such a
discretization\vadjust{\goodbreak}
rule is
characterized by a
pair of positive \mbox{$\mathbb F$-adapted} c\`adl\`ag
processes $(\overline a_t)_{t\geq0}$ and $(\underline a_t)_{t\geq
0}$, and the discretization dates are then defined by \eqref
{discrule}.

%re1 #&#
\begin{remark}
Consider the discretization rules $A = (\underline a,\overline a)$ and
$B = (k\underline a, k\overline a)$ with $k>0$. These two strategies
satisfy $\overline{\mathcal E}^A (C) = \overline{\mathcal E}^B (C)$
for all $C>0$. Therefore, the optimal strategies will be determined up
to a multiplicative constant.
\end{remark}

%pa2.subsection.subsubsection.2 #&#
\textit{Assumptions on the processes $X$ and $Y$}.
Our first main result describing the behavior of the
error functional will be obtained under the assumptions {$(\mathit{HY})$,
$(\mathit{HX})$ and $(\mathit{HX}^1_{\mathrm{loc}})$} stated below.
\begin{longlist}[($\mathit{HX}^\rho_{\mathrm{loc}}$)]
\item[$(\mathit{HY})$] We assume that the process $Y$ is an $\mathbb F$-local
martingale, whose predictable
quadratic variation satisfies $\langle Y \rangle_t = \int_0^t A_s \,ds$,
where the process $(A_t)$ is c\`adl\`ag and locally bounded.

%In the nonmartingale case, this term does not necessarily dominate
%the others (consider e.g., the case when $X_t = t$).

\item[$(\mathit{HX})$] The process $X$ is a semimartingale
defined via the stochastic representation
%
%e13 #&#
%
\begin{eqnarray}\label{defX.eq}
X_t&=& X_0 + \int_0^t
b_s \,ds + \int_0^t \int
_{|z|\leq1} z(M-\mu) (ds\times dz)
\nonumber
\\[-8pt]
\\[-8pt]
\nonumber
&&{} + \int_0^t
\int_{|z|>1} z M(ds\times dz),
\end{eqnarray}
where $M$ is the jump measure of $X$, and $\mu$ is its predictable
compensator, absolutely continuous with respect to the Lebesgue
measure in time, $\mu(dt\times dz) = dt \times\mu_t(dz)$, where the
kernel $\mu_t(dz)$ is such that for some $\alpha\in(1,2)$ there exist
positive c\`adl\`ag
processes $(\lambda_t)$ and $(\widehat K_t)$ and constants $c_+\geq0$
and $c_-\geq0$ with
$c_++c_->0$ and, almost surely for all $t\in[0,T]$,
%
%e14 #&#
%e15 #&#
%
\begin{eqnarray}\qquad
x^\alpha\mu_t((x,\infty)) &\leq&\widehat
K_t \quad\mbox{and}\quad x^\alpha\mu_t((-\infty,-x)
) \leq\widehat K_t \qquad\mbox{for all $x>0$;}\label{hxbnd}
\\
\label{hxalpha}
x^\alpha\mu_t((x,\infty)) &\to& c_+
\lambda_t \quad\mbox{and}\quad x^\alpha\mu_t((-
\infty,-x)) \to c_- \lambda_t
\nonumber
\\[-8pt]
\\[-8pt]
\eqntext{\mbox{when } x\to0.}
\end{eqnarray}

\item[($\mathit{HX}^\rho_{\mathrm{loc}}$)] There exists a L\'evy measure
$\nu(dx)$ such that, almost surely, for all $t$, the kernel
$\mu_t(dz)$ is absolutely continuous with respect to $\lambda_t
\nu(dz)$,
%
%e16 #&#
%
\begin{equation}
\mu_t(dz) = K_t(z) \lambda_t
\nu(dz)\label{ascomp}
\end{equation}
for a random function ${K_t(z)>0}$.
Moreover, there exists an increasing sequence of stopping times $(\tau
_n)$ with
$\tau_n \to T$ such that for every $n$,
%
%e17 #&#
%
\begin{equation}
% dt<K_n,
\int_0^{\tau_n} \int
_{\mathbb R} \bigl|\sqrt{K_t(z)}-1\bigr|^{2\rho} \nu(dz)
\,dt<{ C_n,}\label{integrK}
\end{equation}
$\frac{1}{ C_n} \leq\lambda_t \leq
{ C_n}$, { $\widehat K_t \leq C_n$}
and $|b_t|\leq{ C_n}$ for $0\leq t \leq\tau_n$
and some constant ${ C_n}>0$.
\end{longlist}

%Eventually, for our second theorem, we will need the following slight
%extension of ($\mathit{HX}^\rho_{\mathrm{loc}}$):

%$$\int_{|z|>x}|z|\nu(dz)\leq Cx^{1-\alpha}.$$

%re2 #&#
\begin{remark}[(Concerning the assumptions on the process $Y$)]
The assumption that $Y$ is a local martingale
greatly simplifies the treatment of quadratic hedging problems in
various settings because it allows us to reduce the problem of
minimizing the global quadratic risk to myopic local risk
minimization. In particular, under this assumption, the
error functional \eqref{errfunc} becomes
\[
\mathcal E(\varepsilon) = E \biggl[\int_0^T
(X_t - X_{\eta^{\varepsilon
}(t)})^2 A_t \,dt \biggr].
\]
While it may be unrealistic to assume that the stock
price process is a local martingale \emph{for computing the hedging
strategy}, in the present study we have a different objective. We are
looking for the asymptotically optimal rule to discretize a \emph{given}
strategy, that is, the rule which minimizes, asymptotically for
large number of discretization dates, the principal term of the
discretization error. In the case of equally spaced discretization
dates, it is known (see \cite{tankov.voltchkova.09} for a proof in the
context of It\^o semimartingales with jumps) that this principal term
does not depend on the drift part of the processes $X$ and
$Y$. We conjecture that the same kind of behavior holds in the context
of random rebalancing dates, which means that the drift terms do not
need to be taken into account when computing asymptotically optimal
discretization rules. Our methodology allows us to determine
asymptotically optimal discretization for a given process $X$, which
may correspond, for example, to a quadratic hedging strategy computed
in the nonmartingale setting.
\end{remark}

%re3 #&#
\begin{remark}[(Concerning the assumptions on the process $X$)]\label{purejumprem}
\begin{itemize}[$-$]
\item[$-$] In this paper, we focus on semimartingales
for which the local martingale part is purely discontinuous, with
the aim of determining the effect of small jumps on the
convergence rate of the discretization error. Therefore, we do not
include a continuous local martingale part in the dynamics of $X$.
Indeed, it would asymptotically dominate
the purely discontinuous part as shown in Proposition \ref{toy} in the \hyperref[app]{Appendix}.
The dynamics of $Y$ can, in principle, include such a continuous
local martingale part, however in the usual financial models, when
$X$ has no continuous local martingale part, this is also the case
for $Y$. Note that from the practical viewpoint, many exponential L\'
evy models popular among academics and practitioners (Variance Gamma,
CGMY, Normal inverse Gaussian etc.) do not include a continuous
diffusion part.

\item[$-$]Assumption $(\mathit{HX})$ defines the structure of the
integrand (hedging strategy) $X$, by saying that the small jumps of
$X$ ressemble those of an $\alpha$-stable process, modulated by a
random intensity process $(\lambda_t)$. This assumption introduces the
fundamental parameters which will appear in our limiting results: the
coefficients $\alpha$, $c_+$ and $c_-$ and the intensity process
$\lambda$. These parameters are determined uniquely up to multiplying
$\lambda$ by a positive constant and dividing $c_+$ and $c_-$ by the
same constant. Note also that these parameters can be estimated from
market data; see \cite{woerner,aj4,fig1,fig2}.

\item[$-$] The parameter $\alpha$ measures the activity of
small jumps of the process $X$. In the case where $X$ is a L\'evy
process, the parameter $\alpha$ coincides with the Blumenthal--Getoor
index of $X$; see \cite{bg61}.

\item[$-$]The assumption $1<\alpha<2$ implies that
$X$ has infinite variation and ensures that the local behavior of
the process is determined by the jumps rather than by the drift
part; see \cite{rosenbaum.tankov.10}. Note that in a recent
statistical study on liquid assets~\cite{aj4}, the jump activity
index defined similarly to our parameter $\alpha$ was estimated
between~$1.4$ and $1.7$. However, this assumption does exclude some
interesting models and other statistical studies find that
this parameter can be smaller than one for certain asset classes
\cite{cont.mancini.09,belomestny.10}.

\item[$-$] The assumption $(\mathit{HX}^\rho_{\mathrm{loc}})$ is a technical
integrability condition. In the sequel, we shall always impose
$(\mathit{HX}^1_{\mathrm{loc}})$ and sometimes also $(\mathit{HX}^\rho_{\mathrm{loc}})$ with
$\rho>1$. The representation \eqref{ascomp} of the compensator $\mu$
of the jump measure of $X$ implies that the jump part of $X$ is
locally equivalent to a time-changed L\'evy process. Indeed,
time-changing the process with a continuous increasing process
$\Lambda_t = \int_0^t \lambda_s \,ds$ has the effect of multiplying the
compensator by $\lambda_t$, and making a change of probability measure
with density given by \eqref{girsdens.eq} has the effect of dividing
the compensator by $K_t(z)$. The objects $\nu(dz)$ and $K_t(z)$ in this
representation are not unique, but they do not
appear in our limiting results.
In particular, it is easy to show that the L\'evy measure
$\nu$ necessarily satisfies a stable-like condition similar to \eqref{hxalpha},
%
%e18 #&#
%
\begin{equation}
x^\alpha\nu((x,\infty)) \to c_+ \quad\mbox{and}\quad x^\alpha
\nu((-\infty,-x)) \to c_- \qquad\mbox{when } x\to0.\label{halpha}
\end{equation}
Indeed, there exists a constant $c>0$ such that
\[
c(\sqrt{f} -1)^2 \geq (f-1)^2 \mathbf1_{|f-1|\leq{1}/{2}} +
|f-1| \mathbf1_{|f-1|>{1}/{2}}\qquad \mbox{for all $f>0$}.
\]
From this simple inequality, and denoting $I_t = \int_{\mathbb R}
(\sqrt{K_t(z)}-1)^2 \nu(dz)$, one can easily deduce, using the
Cauchy--Schwarz inequality that for another
constant~$C$,
\[
\biggl\llvert \int_x^\infty\nu(dz) - \int
_x^\infty K_t(z)\nu(dz) \biggr\rrvert
\leq C I_t + C \biggl\{\int_x^\infty
\nu(dz) \biggr\}^{{1}/ {2}} I_t^{{1}/{2}},
\]
and also that
\[
\biggl\llvert \biggl(\int_x^\infty\nu(dz)
\biggr)^{{1}/{2}} - \biggl(\int_x^\infty
K_t(z)\nu(dz) \biggr)^{{1}/{2}} \biggr\rrvert \leq C
I_t
\]
for yet another constant $C$. By \eqref{integrK}, under
$(\mathit{HX}^1_{\mathrm{loc}})$, $I_t<\infty$ for almost all $t$. For any such
$t$, we
can multiply the above inequality with $x^{\alpha/2}$ and take the
limit $x\to0$; we then get
\[
\lim_{x\to0} x^{\alpha} \int_x^\infty
\nu(dz) = \lim_{x\to0} x^{\alpha} \int_x^\infty
K_t(z)\nu(dz),
\]
but the latter limit is equal to $c_+$ by assumption \eqref{hxalpha}.
Moreover, it is always possible with no loss of generality to choose
$\nu$
so that it also satisfies
%
%e19 #&#
%
\begin{equation}
x^\alpha\nu((x,\infty)) + x^\alpha\nu((-
\infty,-x)) \leq C\label{hbnd}
\end{equation}
for some constant $C<\infty$ and all $x>0$. Indeed, by property
\eqref{halpha}, it is enough to show this for all $x\geq\varepsilon$
with some $\varepsilon>0$. But for this, it is enough to take
\[
K_t(z) = \frac{\hat K_t}{\lambda_t}\qquad \mbox{for } |z|\geq \varepsilon
\]
and use \eqref{hxbnd}. Such a choice clearly does not violate
condition \eqref{integrK}.
In the sequel we shall assume that $\nu$
has been chosen in such a way.
\end{itemize}
\end{remark}

%
%ex1 #&#
\begin{example}
In applications, the process $X$ is often defined as solution to a
stochastic differential equation rather than through its
semimartingale characteristics. We now give an example of an SDE which
satisfies our assumptions. Let $X$ be
the solution of an SDE driven by a Poisson random measure
%
%e20 #&#
%
\begin{eqnarray}\label{sde.eq}
X_t &=& X_0 + \int_0^t
\bar b_s \,ds + \int_0^t \int
_{|z|\leq1} \gamma_s(z) \tilde N(ds \times dz)
\nonumber
\\[-8pt]
\\[-8pt]
\nonumber
&&{} + \int
_0^t \int_{|z|> 1}
\gamma_s(z) N(ds \times dz),
\end{eqnarray}
where $N$ is a Poisson random measure with intensity measure $dt
\times\bar\nu(dz)$, $\tilde N$ is the corresponding compensated measure,
and $\gamma\dvtx  [0,T] \times\Omega\times\mathbb R \to\mathbb R$ is
a predictable random function.

%pr1 #&#
\begin{proposition}\label{sde.prop}
Assume that $\bar\nu$ is a L\'evy
measure which has a compact support $U$ such that $0\in  \operatorname{int}
U$ and admits a density also
denoted by $\bar\nu(x)$, which is continuous outside any neighborhood of
zero and is such that
%
%e21 #&#
%
\begin{eqnarray}\label{stable.strong}
x^{\alpha+1} \bar\nu(x) = \alpha c_+ + O(x) \quad\mbox{and}\quad x^{\alpha+1}\bar
\nu(-x) = \alpha c_- + O(x)
\nonumber
\\[-8pt]
\\[-8pt]
\eqntext{\mbox{when } x\downarrow0}
\end{eqnarray}
for some $\alpha\in(1,2)$ and constants $c_+> 0$ and
$c_-> 0$.

Suppose furthermore that for all $\omega\in\Omega$ and $t\in
[0,T]$, $\gamma_t(z)$ is twice differentiable with respect to
$z$, $\gamma'_t(z)>0$ for all $z\in U$, $\gamma_t(0) = 0$, and
there exists an increasing sequence of stopping\vadjust{\goodbreak} times $(\tau_n)$
with $\tau_n \to T$ and a sequence of positive constants $(C_n)$ with
$C_n<\infty$ for all $n$, such that for every $n$, almost surely,
%
%e22 #&#
%
\begin{eqnarray}\label{extracond}
|b_t|\leq C_n, \qquad\frac{1}{C_n} \leq
\gamma'_t(z) \leq C_n \quad\mbox{and}\quad\bigl |
\gamma^{\prime\prime}_t(z)\bigr|<C_n
\nonumber
\\[-8pt]
\\[-8pt]
\eqntext{\mbox{for all } 0\leq t
\leq\tau_n, z\in U.} %&\nu(\gamma^{-1}_t(z)) \leq\frac{C_n}{ |z|^{1+\delta}}, |z|>
\end{eqnarray}
Then the process $X$ satisfies the assumption $(\mathit{HX})$ with $\lambda_t =
\gamma'_t(0)^\alpha$ and the assumption $(\mathit{HX}^\rho_{\mathrm{loc}})$ for
all $\rho\geq1$.\vspace*{-1pt}
\end{proposition}
The proof of this result is given in Appendix \ref{proof1}.\vspace*{-1pt}
\end{example}
%

%pa2.subsection.subsubsection.3 #&#
\textit{Assumptions on the discretization rules}.
Our first main result (asymptotics of the error functional) requires
the following assumptions on the discretization rule~$(\underline a,
\overline a)$:
\begin{longlist}[$(\mathit{HA}_{\mathrm{loc}})$]
\item[$(\mathit{HA})$]{The integrability condition}
\[
E \biggl[\sup_{0\leq s \leq T}\max(\underline a_s,
\overline a_s)^2\int_0^T
A_t \,dt \biggr] <\infty.
\]
\item[$(\mathit{HA}_{\mathrm{loc}})$]
There exists an increasing sequence of stopping times $(\tau_n)$ with
$\tau_n \to T$ such that for every $n$,
$\frac{1}{ C_n} \leq\underline a_t,\overline a_t \leq
{ C_n}$ for $0\leq t \leq\tau_n$
and some constant ${ C_n}>0$.

To obtain our second main result concerning the behavior
of the cost functional, we shall need the following additional
technical assumptions:
\item[$(\mathit{HA}_2)$] For some $\delta\in(0,1)$ with $\beta(1+\delta
)<\alpha$,
\begin{eqnarray*}
&&E \biggl[\sup_{0\leq s\leq T}\bigl(\max\bigl\{\underline
a_s^{\beta
-1},\overline a_s^{\beta-1}\bigr
\}^{1+\delta}+\max\bigl\{\underline a_s^{(1+\delta)\beta
-1},\overline
a_s^{(1+\delta)\beta-1}\bigr\}\bigr)\int_0^T
|b_s|^{1+\delta}\,ds \biggr]
\\
&&\qquad{} + E \biggl[\sup_{0\leq s
\leq T}\max\{\underline a_s,
\overline a_s\}^{(\beta\vee(2-\alpha))(1+\delta)}\min\{\underline a_s,
\overline a_s\}^{((\beta-2)\wedge(-\alpha))(1+\delta)}\\
&&\hspace*{249pt}{}\times \int_0^T
\widehat {K}_s^{1+\delta}\,ds \biggr] <\infty.%\label{intcond}
\end{eqnarray*}
\item[$(\mathit{HA}_2')$]
For some $\delta\in(0,1)$,
\begin{eqnarray*}
&&E \biggl[\sup_{0\leq s\leq T}\min(\underline
a_s,\overline a_s)^{-\alpha(1+\delta)} \int
_0^T \widehat{K}_t^{1+\delta}
\,dt
\\
&&\quad{}+\sup_{0\leq s\leq T}\min(\underline a_s,\overline
a_s)^{-1-\delta} \int_0^T
|b_t|^{1+\delta} \,dt \biggr] <\infty.
\end{eqnarray*}
\end{longlist}

%re4 #&#
\begin{remark}
Condition $(\mathit{HA}'_2)$ replaces condition $(\mathit{HA}_2)$ in the case
$\beta=0$.
For given $\beta$ and given processes $X$ and $Y$, we shall call a
discretization rule $(\underline a, \overline a)$ satisfying
assumptions $(\mathit{HA})$, $(\mathit{HA}_{\mathrm{loc}})$ and $(\mathit{HA}_2)$ (if $\beta>0$) or
assumptions $(\mathit{HA})$, $(\mathit{HA}_{\mathrm{loc}})$ and $(\mathit{HA}'_2)$ (if $\beta=0$) \emph{an
admissible discretization rule}.\looseness=-1\vadjust{\goodbreak}
\end{remark}

%s3 #&#
\section{Main results}
In this section, we first characterize the asymptotic behavior of the
error and cost functionals for small $\varepsilon$. From these results
we then derive the asymptotically optimal discretization strategies
using Lemma
\ref{errorcost.lm}.
%s3.1 #&#
\subsection{Asymptotic behavior of the error and cost
functionals}\label{main.thm}
%
%th1 #&#
\begin{theorem}\label{err.thm}
$\!\!\!\!$Under assumptions {$(\mathit{HY})$, $(\mathit{HX})$, $(\mathit{HX}^1_{\mathrm{loc}})$,
$(\mathit{HA})$ and $(\mathit{HA}_{\mathrm{loc}})$,}
%
%e23 #&#
%
\begin{equation}
\lim_{\varepsilon\to0} \varepsilon^{-2} \mathcal E(\varepsilon)
= E \biggl[\int_0^T {A_t}
\frac{f(\underline a_t, \overline
a_t)}{g(\underline a_t, \overline a_t)} \,dt \biggr],\label{err.eq}
\end{equation}
where, for $\underline a, \overline a \in(0,\infty)$,
\[
f(\underline a, \overline a) = E \biggl[\int_0^{\tau^*}
\bigl(X^*_t\bigr)^2 \,dt \biggr],\qquad g(\underline a, \overline
a) = E\bigl[\tau^*\bigr]
\]
with $\tau^* = \inf\{t\geq0\dvtx X^*_t \notin(-\underline a, \overline a)
\}$, where $X^*$ is a strictly $\alpha$-stable process with L\'evy
density
\[
\nu^*(x) = \frac{c_+ 1_{x>0} + c_- 1_{x<0}}{|x|^{1+\alpha}},\qquad x\neq0,
\]
and the constants $c_-$ and $c_+$ are defined in assumption
$(\mathit{HX})$ [equation \eqref{hxalpha}].
\end{theorem}

%th2 #&#
\begin{theorem}\label{cost.thm}
We use the notation of Theorem \ref{err.thm}.
\begin{longlist}[(ii)]
\item[(i)] Let assumptions {$(\mathit{HY})$, $(\mathit{HX})$,
$(\mathit{HX}^{1}_{\mathrm{loc}})$, $(\mathit{HA})$, $(\mathit{HA}_{\mathrm{loc}})$ and $(\mathit{HA}_2')$}
be satisfied. Then
%
%e24 #&#
%
\begin{equation}
\lim_{\varepsilon\to0} \varepsilon^{\alpha} \mathcal
C^0(\varepsilon) = E \biggl[\int_0^T
\frac{{\lambda_t}}{g(\underline a_t, \overline a_t)} \,dt \biggr].\label{cost0.eq}
\end{equation}
\item[(ii)] Let $\beta\in(0,\alpha)$, and assume that {$(\mathit{HY})$,
$(\mathit{HX})$, $(\mathit{HX}^{1}_{\mathrm{loc}})$,
$(\mathit{HX}^{\rho}_{\mathrm{loc}})$ (for some $\rho>
\frac{\alpha}{\alpha-\beta}\vee2$), $(\mathit{HA})$, $(\mathit{HA}_{\mathrm{loc}})$ and
$(\mathit{HA}_2)$} hold true.
Then
%
%e25 #&#
%
\begin{equation}
\lim_{\varepsilon\to0} \varepsilon^{\alpha-\beta} \mathcal
C^{ \beta
}(\varepsilon) = E \biggl[\int_0^T
{\lambda_t} \frac{u^\beta(\underline a_t, \overline
a_t)}{g(\underline a_t, \overline a_t)} \,dt \biggr],\label{cost.eq}
\end{equation}
where
\[
u^\beta(\underline a, \overline a) = E\bigl[\bigl|X^*_{\tau^*}\bigr|^\beta
\bigr]<\infty.
\]
\end{longlist}
\end{theorem}
%

%
%re5 #&#
\begin{remark}
Theorems \ref{err.thm} and \ref{cost.thm} enable us to
apply Lemma \ref{errorcost.lm} and conclude that \emph{for any
admissible discretization rule based on hitting times}, the error
functional for fixed cost behaves, for large costs, as
\[
\overline{\mathcal E}(C) \sim C^{-{2}/{(\alpha-\beta)}} E \biggl[\int
_0^T {A_t} \frac{f(\underline a_t, \overline a_t)}{g(\underline a_t,
\overline a_t)} \,dt
\biggr] E \biggl[\int_0^T {
\lambda_t} \frac{u^\beta(\underline a_t,
\overline a_t)}{g(\underline a_t, \overline a_t)} \,dt \biggr]^{{2}/{(\alpha-\beta)}}.
\]
When the cost is equal to the expected number of rebalancings
($\beta=0$), the error converges to zero at the rate $C^{-{2}/{\alpha}}$.
On the other hand, for equidistant rebalancing dates,
under sufficient regularity, the $L^2$ discretization error of the
quadratic hedging strategy in exponential L\'evy models is
\emph{inversely proportional} to the number of rebalancings; see \cite
{broden.tankov.09}. This means that while in diffusion models,
asymptotically optimal hedging reduces the error without modifying
the rate at which the error decreases with the number of rebalancings
[cf. equations \eqref{diffeff} and \eqref{diffunif}],
in pure jump models, any discretization based on hitting times, and a
fortiori the optimal discretization, also
improves the rate of convergence.
\end{remark}
%

%s3.2 #&#
\subsection{Application: Computing the optimal barriers}
\label{barrier.sec}
In this section, we suppose that the assumptions of Theorem
\ref{cost.thm} [part (i) or (ii), depending on $\beta$] are
satisfied. {In view of Lemma \ref{errorcost.lm}, we shall use the
following definition of an asymptotically optimal discretization
rule.
%
%de2 #&#
\begin{definition}
A discretization rule
$(\underline a,\overline a)$ is said to be asymptotically optimal if it
is admissible, and for any other admissible rule $(\underline
a',\overline
a')$,
%
%e26 #&#
%e27 #&#
%
\begin{eqnarray}
\label{assoptstab}&& E \biggl[\int_0^T
{A_t} \frac{f(\underline a_t, \overline
a_t)}{g(\underline a_t, \overline a_t)} \,dt \biggr] E \biggl[\int_0^T
{\lambda_t} \frac{u^\beta(\underline a_t,
\overline a_t)}{g(\underline a_t, \overline a_t)} \,dt \biggr]^{{2}/{(\alpha-\beta)}}
\nonumber
\\[-8pt]
\\[-8pt]
\nonumber
&&\qquad\leq E \biggl[\int_0^T {A_t}
\frac
{f(\underline a'_t, \overline a'_t)}{g(\underline a'_t, \overline a'_t)} \,dt \biggr] E \biggl[\int_0^T
{\lambda_t} \frac{u^\beta(\underline a'_t,
\overline a'_t)}{g(\underline a'_t, \overline a'_t)} \,dt \biggr]^{{2}/{(\alpha-\beta)}}.
\end{eqnarray}
\end{definition}
The following result simplifies the characterization of
such rules.

%pr2 #&#
\begin{proposition}\label{optbar.prop}
Let $(\underline a,\overline a)$ be an admissible discretization rule, and
assume that there exists $c>0$ such that for any other admissible rule
$(\underline a',\overline
a')$,
%
%e28 #&#
%
\begin{equation}
\label{assoptlagrange} {A_t} \frac{f(\underline
a_t, \overline a_t)}{g(\underline a_t, \overline a_t)} + c {\lambda_t}
\frac{u^\beta(\underline a_t, \overline a_t)}{g(\underline
a_t, \overline a_t)} \leq{A_t} \frac{f(\underline
a'_t, \overline a'_t)}{g(\underline a'_t, \overline a'_t)} + c {
\lambda_t} \frac{u^\beta(\underline a'_t, \overline a'_t)}{g(\underline
a'_t, \overline a'_t)}
\end{equation}
a.s. for all $t\in[0,T]$. Then the rule $(\underline a,\overline a)$
is asymptotically optimal.
\end{proposition}
\begin{pf}
By the nature of assumptions $(\mathit{HA})$, $(\mathit{HA}_{\mathrm{loc}})$ and $(\mathit{HA}_2)$
[resp., $(\mathit{HA}'_2)$], for all $\kappa>0$,
the rule $(\kappa\underline a, \kappa\overline a)$ is admissible. In
addition, by the
scaling property of strictly stable processes,
\begin{eqnarray*}
f(\kappa\underline a_t, \kappa\overline a_t) &=&
\kappa^{2+\alpha} f(\underline a_t,\overline a_t), \qquad g(
\kappa\underline a_t, \kappa \overline a_t) =
\kappa^{\alpha} g(\underline a_t,\overline a_t),\\
u^\beta(\kappa\underline a_t, \kappa\overline
a_t) &=& \kappa^{\beta} u^\beta(\underline
a_t,\overline a_t).
\end{eqnarray*}
Using these identities in the left-hand side of \eqref{assoptlagrange}
and the fact that \eqref{assoptlagrange} holds for any $(\underline
a',\overline a')$,
integrating both sides and taking the expectation, we get
\begin{eqnarray*}
&& E \biggl[\int_0^T {A_t}
\frac{f(\underline
a_t, \overline a_t)}{g(\underline a_t, \overline a_t)}\,dt \biggr]+ c \kappa^{\alpha-\beta} E \biggl[\int
_0^T {\lambda_t}
\frac{u^\beta(\kappa\underline a_t, \kappa
\overline a_t)}{g(\kappa\underline
a_t, \kappa\overline a_t)}\,dt \biggr]
\\
&&\qquad\leq E \biggl[\int_0^T {A_t}
\frac
{f(\kappa'\underline
a'_t, \kappa'\overline a'_t)}{g(\kappa'\underline a'_t, \kappa
'\overline a'_t)}\,dt \biggr] + c E \biggl[\int_0^T
{\lambda_t} \frac{u^\beta(\kappa'\underline a'_t, \kappa
'\overline a'_t)}{g(\kappa'\underline
a'_t, \kappa'\overline a'_t)}\,dt \biggr].
\end{eqnarray*}
Under the assumptions of Theorems \ref{err.thm} and \ref{cost.thm}, all
expectations above are finite. Indeed, the limiting error functional
is finite by assumption $(\mathit{HA})$ since clearly $f(\underline a, \overline
a)\leq\max(\underline a, \overline a)^2 g(\underline a, \overline
a)$. The finiteness of the limiting cost functional is shown by
applying Lemma \ref{overshoot.lm} to the limiting strictly stable
process to obtain a bound on the function $u^\beta$ and then using
assumption $(\mathit{HA}_2)$ or $(\mathit{HA}'_2)$.

Now, choose $\kappa$ so that
\[
E \biggl[\int_0^T {\lambda_t}
\frac{u^\beta(\kappa\underline a_t, \kappa
\overline a_t)}{g(\kappa\underline
a_t, \kappa\overline a_t)}\,dt \biggr] = 1 \quad\Rightarrow\quad \kappa= E \biggl[\int
_0^T {\lambda_t}
\frac{u^\beta(\underline a_t,
\overline a_t)}{g(\underline
a_t, \overline a_t)}\,dt \biggr]^{{1}/{(\beta-\alpha)}}
\]
and $\kappa'$ so that
\begin{eqnarray*}
&& E \biggl[\int_0^T {\lambda_t}
\frac{u^\beta(\kappa'\underline a'_t, \kappa
'\overline a'_t)}{g(\kappa'\underline
a'_t, \kappa'\overline a'_t)}\,dt \biggr] = \kappa^{\alpha-\beta}\\
&&\quad\Rightarrow\quad\displaystyle\kappa' = \frac{1}{\kappa} E \biggl[\int_0^T
{\lambda_t} \frac{u^\beta(\underline a'_t, \overline a'_t)}{g(\underline
a'_t, \overline a'_t)}\,dt \biggr]^{{1}/{(\alpha-\beta)}}.
\end{eqnarray*}
This yields
\[
E \biggl[\int_0^T {A_t}
\frac{f(\underline
a_t, \overline a_t)}{g(\underline a_t, \overline
a_t)}\,dt \biggr] \leq\bigl(\kappa'\bigr)^2 E
\biggl[\int_0^T {A_t}
\frac
{f(\underline
a'_t, \overline a'_t)}{g(\underline a'_t, \overline a'_t)}\,dt \biggr].
\]
Substituting the expression for $\kappa'$, we finally obtain \eqref{assoptstab}.
\end{pf}

The above result shows that we may look for optimal barriers as
$\underline a$ and $\overline a$ as minimizers of
%
%e29 #&#
%
\begin{equation}
\min \biggl\{{A_t} \frac{f(\underline
a_t, \overline a_t)}{g(\underline a_t, \overline a_t)} + c {\lambda_t}
\frac{u^\beta(\underline a_t, \overline a_t)}{g(\underline a_t,
\overline a_t)} \biggr\},\label{lagmin}
\end{equation}
provided that the resulting $\underline a_t$ and $\overline a_t$
are admissible. Moreover if $(\underline
a,\overline a)$ is the solution of \eqref{lagmin}, then the scaling
property shows that the solution of
\[
\min \biggl\{{A_t} \frac{f(\underline
a_t, \overline a_t)}{g(\underline a_t, \overline a_t)} + c' {\lambda
_t} \frac{u^\beta(\underline a_t, \overline a_t)}{g(\underline a_t,
\overline a_t)} \biggr\}
\]
is given by $(\kappa\underline a, \kappa\overline a)$ with $\kappa=
(c'/c )^{{1}/{(\alpha-\beta+2)}}$. If $c'>c$, then
$\kappa>1$, resulting in a smaller cost functional and a bigger error
functional. Therefore, in practice $c$ may be chosen by the trader
depending on the maximum acceptable cost:} the bigger $c$, the smaller
will be the cost of the strategy and, consequently the bigger its
error.

The functions $f$, $g$ and $u$ appearing above must in general
be computed numerically. However, when {the constants $c_+$
and $c_-$ in \eqref{hxalpha} are equal}, which
corresponds for example to the CGMY model very popular in practice
\cite{finestructure}, the results are completely explicit, as will be
shown in the next paragraph.
%s3.3 #&#
\subsection{Locally symmetric L\'evy measures}
\label{localsym.sec}
In this section we discuss a case important in applications,
when the asymptotically optimal strategy can be computed explicitly in
terms of $A$ and $\lambda$.

%pr3 #&#
\begin{proposition}\label{sym.prop}
Let the cost functional be of the form \eqref{costfunc} with $\beta
\in[0,1]$. Let the processes $X$ and $Y$ satisfy the assumptions
$(\mathit{HY})$, $(\mathit{HX})$ with $c_+=c_-$, $(\mathit{HX}^{1}_{\mathrm{loc}})$ and $(\mathit{HX}^{\rho}_{\mathrm{loc}})$ with
$\rho>\frac{\alpha}{\alpha-\beta}\vee2$ (if $\beta>0$). Assume that
the processes
$A$, $b$ and $\lambda$ satisfy the following integrability conditions for
some $\delta>0$:
\begin{eqnarray*}
E \biggl[ \biggl(\sup_{0\leq t \leq T} \frac{\lambda_t}{A_t}
\biggr)^{{2}/{(2+\alpha-\beta)}} \int_0^T A_t
\,dt \biggr]&<&\infty,
\\
E \biggl[ \biggl(\inf_{0\leq t \leq T}
\frac{\lambda_t}{A_t} \biggr)^{{(1+\delta)(\beta-\alpha)}/{(2+\alpha
-\beta)}} \int_0^T
\widehat{K}_t^{1+\delta} \,dt \biggr]&<&\infty,
\end{eqnarray*}
and, if $\beta=1$,
\[
E \biggl[ \biggl(\sup_{0\leq t \leq T} \frac{\lambda_t}{A_t}
\biggr)^{\delta} \int_0^T
|b_t|^{1+\delta} \,dt \biggr]<\infty,
\]
or, if $\beta<1$,
\[
E \biggl[ \biggl(\inf_{0\leq t \leq T} \frac{\lambda_t}{A_t}
\biggr)^{(\beta-1)(1+\delta)} \int_0^T
|b_t|^{1+\delta} \,dt \biggr]<\infty.
\]
Then the strategy given by
\[
\underline a_t = \overline a_t = c \biggl(
\frac{\lambda_t}{A_t} \biggr)^{{1}/{(2+\alpha-\beta)}}
\]
is asymptotically optimal.
\end{proposition}
\begin{pf}
The fact that $X$ satisfies $(\mathit{HX})$ with $c_+=c_-$ means that the
limiting process $X^*$ is a symmetric stable process.
%with characteristic function $E[e^{iuX^*_t}] = e^{t\sigma|u|^\alpha}$.
Let $(\underline a,\overline a)$ be an {admissible} discretization
rule. With a change of notation $a_t:= \frac{\underline a_t +
\overline a_t}{2}$ and $\theta_t = \frac{\overline a_t - \underline
a_t}{\overline a_t + \underline a_t}$ and using\vadjust{\goodbreak} the results from
Appendix \ref{appa} [Proposition \ref{squarefunc}, equations \eqref{tauab} and
\eqref{overshootab}], we can compute
\begin{eqnarray*}
\hspace*{-4pt}&&\frac{f(\underline a_t, \overline a_t)}{g(\underline a_t, \overline
a_t)} = \frac{\alpha}{(\alpha+2)(\alpha+1)}a_t^2 \bigl(1 +
\theta_t^2(1+\alpha) \bigr),
\\
\hspace*{-4pt}&&\frac{u^\beta(\underline a_t, \overline a_t)}{g(\underline a_t,
\overline
a_t)} \\
\hspace*{-4pt}&&\qquad=\frac{\sigma\Gamma(1+\alpha)\sin{\pi\alpha}/{2}}{\pi}
\\
\hspace*{-4pt}&&\qquad\quad{} \times\int_0^\infty z^{-\alpha/2}
(z+2a_t )^{-\alpha/2} \bigl(\bigl|z+a_t(1+
\theta_t)\bigr|^{\beta-1} + \bigl|z+a_t(1-\theta_t)\bigr|^{\beta-1}
\bigr)\,dz.
\end{eqnarray*}
For fixed $a_t$, both ratios are minimal when $\theta=0$
(for the second functional this follows from the convexity of the
function $x\mapsto x^{\beta-1}$ for $x\geq0$ and $\beta\leq
1$). Moreover,
from the structure of assumptions $(\mathit{HA})$, $(\mathit{HA}_{\mathrm{loc}})$ and $(\mathit{HA}_2)$
[resp., $(\mathit{HA}'_2)$], it is clear that the strategy obtained by taking
$\theta=0$, that is, the strategy $(a,a)$ is also
admissible. Therefore, the asymptotically optimal strategy, if it
exists, will be symmetric in this case. By the same arguments as in
the previous section, we can show that the optimal strategy, if it
exists, minimizes
\[
A_t \frac{f(a_t,a_t)}{g(a_t,a_t)} + c \lambda_t \frac{u^\beta
(a_t,a_t)}{g(a_t,a_t)}
\]
for each $t$. Plugging in the explicit expressions computed above, we
see that this functional is minimized by
\[
a_t = c \biggl( \frac{\lambda_t}{A_t} \biggr)^{{1}/{(2+\alpha-\beta)}}
\]
for a different constant $c$.
By the assumptions of the proposition, this strategy is admissible, which
completes the proof.
\end{pf}
%

%For this example, assume that the cost is simply given by the expected
%number of rebalancings: $\mathcal
%C^\varepsilon_T = E[N^\varepsilon_T]$. A similar analysis to that of
%the error functional then yields:
%$$
%= E\left[\int_0^T \frac{\lambda_tdt}{g(\underline a_t, \overline
% a_t)}\right] = c E\left[\int_0^T \lambda_t a_t^{-\alpha}(1-
%$$
%where $c$ is a constant.
%Similarly, the cost functional satisfies

%s3.4 #&#
\subsection{Exponential L\'evy models}
\label{explevy.sec}
{In this section we treat the case when the process $Y$ (the asset
price or the integrator) is the stochastic exponential of a L\'evy
process. More
precisely, throughout this section we assume that
\[
Y_t = Y_0 + \int_0^t
Y_{s-} \,dZ_s,
\]
where $Z$ is a martingale L\'evy process with no diffusion part and with
L\'evy measure $\nu$ which has a compact support $U\in(-1,\infty)$
with $0 \in\operatorname{int}U$ and admits a density $\bar\nu(x)$ which
is continuous outside any neighborhood of zero and satisfies~\eqref{stable.strong}. From the martingale property and the
boundedness of jumps of $Z$, it follows immediately that assumption
$(\mathit{HY})$ is satisfied with $A_t = Y_t^2 \int_{\mathbb R} z^2 \bar
\nu(dz)$. For the choice of the integrator $X$ we consider two
examples corresponding to the discretization of hedging strategies on
one hand and to the discretization of optimal investment policies on
the other hand.

%ex2 #&#
\begin{example}[(Discretization of hedging strategies)]
In this example we assume that the integrand $X$ (the hedging strategy)
is a
deterministic function of $Y$, which is indeed the case for classical
strategies (quadratic hedging, delta hedging) and European contingent
claims in exponential L\'evy
models; see \cite{kallsen.hubalek.al.06,broden.tankov.09}.

%pr4 #&#
\begin{proposition} \label{levy.prop}
Let $X_t = \phi(t,Y_t)$ with $\phi(t,y) \in C^{1,2}([0,T)\times\mathbb
R)$ such that for all $\bar Y>0$ and $T^*\in[0,T)$,
\[
\min_{(t,y)\in[0,T^*]\times[-\bar Y,\bar Y] } \frac{\partial
\phi(t,y)}{\partial y} >0.
\]
Then, assumptions $(\mathit{HY})$, $(\mathit{HX})$ and $(\mathit{HX}^{\rho}_{\mathrm{loc}})$ (for
all $\rho\geq1$) are satisfied with
\[
b_t = \frac{\partial\phi}{\partial t} (s,Y_s)+ \frac{\partial\phi}{\partial y}
(s,Y_{s})Y_{s}\int_{|z|>1}z\bar\nu(dz)\quad
\mbox{and}\quad \lambda_t = \biggl(Y_t \frac{\partial\phi}{\partial
y}(t,Y_t)
\biggr)^\alpha.
\]
Assume additionally that the function $\phi$ is such that the
integrability conditions of Proposition \ref{sym.prop} are satisfied
for some $\delta>0$.
Then the strategy given by
\[
\underline a_t = \overline a_t = c \biggl(
\frac{\partial\phi(t,Y_t)}{\partial
y} \biggr)^{{\alpha}/{(2+\alpha-\beta)}} Y_t^{{(\alpha-2)}/{(\alpha-\beta+2)}}
\]
is asymptotically optimal.
\end{proposition}
\begin{pf}
Applying It\^o's formula to $\phi(t,Y_t)$, we get
\begin{eqnarray*}
X_t &=& \phi(0,Y_0) + \int_0^t
b_s \,ds+ \int_0^t \int
_{|z|\leq1} \gamma_s(z) \tilde N(ds\times dz)\\
&&{} + \int
_0^t \int_{|z|> 1}
\gamma_s(z) N(ds\times dz)
\end{eqnarray*}
with $\gamma_t(z) = \phi(t,Y_{t} (1+z)) -
\phi(t,Y_{t})$, which means that we can apply Proposition~\ref{sde.prop}. The local boundedness conditions required by this
proposition follow from the local boundedness of $Y$ and the
continuity of the derivatives of $\phi$. The second statement is a
direct corollary of Proposition \ref{sym.prop}.
\end{pf}

%re6 #&#
\begin{remark}
Using the Cauchy--Schwarz inequality and the fact that $Y$ admits all
moments (because $Z$ has bounded jumps), one can show that the
following more compact condition implies the integrability conditions
of Proposition \ref{sym.prop}: for some $\delta>0$,
\begin{eqnarray*}
&&E \Bigl[ \Bigl(\sup_{x\in U, 0\leq t \leq T} \phi'_y
\bigl(t,Y_t(1+x)\bigr)+\sup_{0\leq t \leq T} \bigl|
\phi'_t(t,Y_t)\bigr| \Bigr)^{2+\delta}\\
&&\hspace*{118pt}{}+
\Bigl(\inf_{0\leq t \leq T} \phi'_t(t,Y_t)
\Bigr)^{-\alpha(2+\delta)} \Bigr]<\infty.
\end{eqnarray*}
This condition can be checked for specific strategies and specific
parametric L\'evy models using the explicit formulas
for the hedging strateigies given in
\cite{kallsen.hubalek.al.06,broden.tankov.09}, but these computations
are out of scope of the present paper.
\end{remark}

%re7 #&#
\begin{remark}
When $\beta=0$ and $\alpha\to2$, we find that
the optimal size of the rebalancing interval is proportional to
the square root of $\frac{\partial\phi(t,Y_t)}{\partial
Y}$ (the gamma), which is consistent with the results of Fukasawa
\cite{fukasawa.09b}, quoted in the \hyperref[sec1]{Introduction}.
\end{remark}
\end{example}

%ex3 #&#
\begin{example}[(Discretization of Merton's portfolio strategy)]
A widely popular portfolio strategy, which was shown by Merton
\cite{merton} to be optimal in the context of power
utility maximization, is the so called constant proportion strategy,
which consists of investing a fixed fraction of one's wealth into
the risky asset. Since the price of the risky asset evolves with time,
the number of units which corresponds to a given proportion varies,
and in practice the strategy must be discretized. Given the importance
of this strategy in applications, it is of interest to compute the
asymptotically optimal discretization rule in this setting.

Assuming zero interest rate, the value
$V_t$ of a portfolio which invests a proportion $\pi$ of the wealth
into the risky asset $Y$ and the rest into the risk-free bank account has
the dynamics
%
%e30 #&#
%
\begin{equation}
V_T = V_0 +\int_0^T
\pi V_{t-} \frac{dY_t}{Y_{t-}}= V_0 + \int
_0^T X_{t-} \,dY_t\qquad
\mbox{with } X_t = \pi\frac{V_t}{Y_t}.\label{merton.eq}
\end{equation}
The following result provides the asymptotically optimal
discretization rule for this integral.

%pr5 #&#
\begin{proposition}\label{merton.prop}
Assume that $U \subset (-\frac{1}{\pi},\infty )$ if $\pi>1$ and
$U\subset (-1, -\frac{1}{\pi} )$ if $\pi<0$. Then the strategy
given by
\[
\underline{a}_t = \overline{a}_t = c
V_t^{{\alpha}/{(2+\alpha-\beta)}} Y_t^{-{(2+\alpha)}/{(2+\alpha-\beta)}}
\]
is asymptotically optimal for the integral \eqref{merton.eq}.
\end{proposition}
\begin{pf}
Applying the It\^o's formula, we find the dynamics of the integrator $X$,
\[
X_t = X_0 + (\pi-1) \int_0^t
\int_{U} \frac{X_{s-}z}{1+z} \tilde N(ds\times dz) + (1-\pi)
\int_0^t\int_{U}
\frac{ X_{s-}z^2}{1+z}\nu(dz) \,ds.
\]
Hence, $X$ can be written in the form of \eqref{sde.eq} with
\[
\gamma_s(z) = \frac{(\pi-1)X_{s-}z}{1+z}\quad \mbox{and}\quad \bar b_s
= (1-\pi) X_{s} \int_{\mathbb R} \biggl\{
\frac{z^2}{1+z} 1_{|z|\leq
1} + z1_{|z|>1} \biggr\} \nu(dz).
\]
Under the assumption of this proposition, the process $X$ does not
change sign, and we can assume without loss of generality that
$(\pi-1)X_s$ is always positive (otherwise all the computations can be
done for the process $-X$). Since $X$ is a stochastic exponential of a
L\'evy process with bounded jumps, it is locally bounded, which means
that by Proposition \ref{sde.prop}, $X$ satisfies the assumption
$(\mathit{HX})$ with
\[
\lambda_t = \gamma'_t(0)^\alpha=
\bigl|(\pi-1)X_{t-}\bigr|^\alpha
\]
and the assumption $(\mathit{HX}^{\rho}_\mathrm{loc})$ for all $\rho\geq
1$. Moreover, since the compensator of the jump measure of $X$ is
absolutely continous with respect to the Lebesgue measure (in time),
we can take $\lambda_t = |(\pi-1)X_{t}|^\alpha$. Also, one can choose
$\widehat K_t = C X_t$ for $C$ sufficiently large in condition
\eqref{hxbnd}.

To check the integrability conditions in Proposition \ref{sym.prop},
observe that the processes $A_t$, $\lambda_t$, $\widehat K_t$ and $b_t$
appearing in these conditions, are powers of stochastic exponentials of
L\'evy
processes with bounded jumps. They can therefore be represented as
ordinary exponentials of (other) L\'evy processes with bounded jumps,
but an exponential of a L\'evy process with bounded jumps admits all
moments, and its maximum on $[0,T]$ also admits all moments; see
Theorem 25.18 in \cite{sato}. Therefore, the integrability conditions
in Proposition \ref{sym.prop} follow by using the Cauchy--Schwarz
inequality, and the proof is completed by an application of this proposition.
\end{pf}
\end{example}
}
%s4 #&#
\section{Proof of Theorem \protect\texorpdfstring{\ref{err.thm}}{1}}\label{proofthm}
\emph{Step} 1. \textit{Reduction to the case of bounded coefficients}.
In the
proofs of Theorems \ref{err.thm} and \ref{cost.thm}, we will replace
the {local boundedness and integrability assumptions of these theorems}
with the following stronger one:
\begin{longlist}[($H'_\rho$)]
\item[($H'_\rho$)] There exists a constant {$B>0$ such that $\frac{1}{B}
\leq\lambda_t, \underline a_t, \overline a_t \leq B$, $|A_t| +
|b_t| + |\widehat K_t|\leq B$ for $0\leq t \leq T$. There exists a L\'evy
measure $\nu(dx)$ such that, almost surely for all $t$, the kernel
$\mu_t(dz)$ is absolutely continuous with respect to $\lambda_t
\nu(dz)$: $\mu_t(dz) = K_t(dz) \lambda_t \nu(dz)$ for a random
function $K_t(z)>0$.} Moreover the process $(Z_t)$ defined by
%
%e31 #&#
%
\begin{equation}
Z_t = \mathcal E \biggl(\int_0^{\cdot}
\bigl( \bigl(K_s(z) \bigr)^{-1}-1 \bigr) (M -
\mu)\label{girsdens.eq} (ds\times dz) \biggr)_t,
\end{equation}
is a martingale and satisfies
\[
E^Q\Bigl[\sup_{0\leq t \leq T} |Z_t|^{-\rho}
\Bigr] <\infty \quad\mbox{and}\quad E\Bigl[\sup_{0\leq t \leq T} Z_t
\Bigr] <\infty,
\]
where $Q$ is the probability measure defined by
\[
\frac{dQ}{dP}\Big|_{\mathcal F_T}:= Z_T.
\]
\end{longlist}
Indeed, we have the following lemma.\vadjust{\goodbreak}
%
%le2 #&#
\begin{lemma}\label{loc.lm}
Assume that \eqref{err.eq} holds under the assumptions {$(\mathit{HY})$, $(\mathit{HX})$}
and $(H'_1)$. Then Theorem \ref{err.thm} holds.
\end{lemma}
\begin{pf}
First, observe that for every $n$,
\begin{eqnarray*}
&&E \biggl[ \biggl\{\int_0^{\tau_n} \int
_{\mathbb R} \bigl( \bigl(K_s(z) \bigr)^{-1}-1
\bigr)^2 M(ds\times dz) \biggr\}^{{1}/{2}} \biggr]
\\
&&\qquad\leq E \biggl[ \biggl\{\int_0^{\tau_n} \int
_{|K_s(z)^{-1}-1|\leq
{1}/{2}} \bigl( \bigl(K_s(z) \bigr)^{-1}-1
\bigr)^2 M(ds\times dz) \biggr\}^{{1}/{2}} \biggr]
\\
&&\qquad\quad{} + E \biggl[ \biggl\{\int_0^{\tau_n} \int
_{|K_s(z)^{-1}-1|>{1}/{2}} \bigl( \bigl(K_s(z) \bigr)^{-1}-1
\bigr)^2 M(ds\times dz) \biggr\}^{{1}/{2}} \biggr].
\end{eqnarray*}
Using the Cauchy--Schwarz inequality for the first term and the fact
that the second integral is a countable sum together with Proposition
II.1.28 in \cite{jacodshiryaev} for the second term, we see that this
last expression is finite since by assumption {$(\mathit{HX}^1_{\mathrm{loc}})$},
\begin{eqnarray*}
&&E \biggl[\int_0^{\tau_n} \int_{|K_s(z)^{-1}-1|\leq{1}/{2}}
\bigl( \bigl(K_s(z) \bigr)^{-1}-1 \bigr)^2
\mu(ds\times dz) \biggr]^{{1}/{2}}
\\
&&\qquad{} + E \biggl[\int_0^{\tau_n} \int
_{|K_s(z)^{-1}-1|>{1}/{2}}\bigl | \bigl(K_s(z) \bigr)^{-1}-1 \bigr|
\mu(ds\times dz) \biggr] < \infty.
\end{eqnarray*}
This implies that the process
\[
L_t = \int_0^t\int
_{\mathbb R} \bigl( \bigl(K_s(z) \bigr)^{-1}-1
\bigr) (M-\mu ) (ds\times dz)
\]
is a local martingale and satisfies $E[[L]^{{1}/{2}}_{T\wedge\tau_n}]
<\infty$ for every $n$; see Definition~II.1.27 in \cite{jacodshiryaev}.
The process $Z_t:= \mathcal E(L)_t$ is then
also well defined, and we take $\sigma_n:= \tau_n \wedge\inf\{t\dvtx Z_t \geq n\}$. Then
\begin{eqnarray*}
\sup_{0\leq t\leq T} Z_{t\wedge\sigma_n} &\leq& n + \bigl|\Delta
Z_{\sigma_n}\bigr|1_{\sigma_n \leq T} \leq n + [Z]^{{1}/{2}}_{\sigma_n
\wedge T} =
n + \biggl(\int_0^{\sigma_n \wedge T} Z_{t-}^2
d[L]_t \biggr)^{{1}/{2}}
\\
&\leq& n + n [L]^{{1}/{2}}_{\sigma_n \wedge T},
\end{eqnarray*}
the last term being integrable. Therefore, we can define a new probability
measure~$Q^n$ via
\[
\frac{dQ^n}{dP}\Big|_{\mathcal F_t} = Z_{t\wedge\sigma_n}.
\]
By Girsanov's theorem (Theorem III.5.24 in \cite{jacodshiryaev}),
$M$ is a random measure with predictable compensator $\mu^{Q^n}:= dt
\times
\lambda_t \nu(dz)$ under $Q^n$ on
$\{t\leq\sigma_n\}$ and
\[
Z^{-1}_{t\wedge\sigma_n} = \mathcal E \biggl(\int_0^\cdot
\bigl(K_s(z)-1\bigr) \bigl(M-\mu^{Q^n}\bigr) (ds\times dz)
\biggr)_{t\wedge\sigma_n}.
\]
Therefore, by similar arguments to above, we can find an increasing
sequence of stopping times $(\gamma_n)$ with $\gamma_n \to T$ and
such that both
\[
E\Bigl[\sup_{0\leq t \leq T} Z_{t\wedge\gamma_n}\Bigr] <\infty\quad \mbox{and}\quad
E^{Q^n}\Bigl[\sup_{0\leq t \leq T} Z^{-1}_{t\wedge\gamma
_n}
\Bigr] <\infty.
\]
Now we define $Y^n_t = Y_{t\wedge\gamma_n}$ and $X^n$ via
equation \eqref{defX.eq} replacing the coefficients $\lambda_t$,
$b_t$ and $K_t(z)$ with $\lambda^n_t:= \lambda_{t\wedge
\gamma_n}$, $b^n_t:= b_{t\wedge\gamma_n}$ and $K^n_t(z) = K_t(z)
1_{t\leq\gamma_n} + 1_{t>\gamma_n}$. Moreover, we define
$\underline a^n_t:= \underline a_{t\wedge\gamma_n}$, $\overline
a^n_t:= \overline a_{t\wedge\gamma_n}$. The stopping times
$T_i^{\varepsilon,n}$ and $\eta^n(t)$ are defined similarly. Note that
$A^n_t:= A_t 1_{t\leq
\gamma_n}$ satisfies $\int_0^t A^n_s \,ds = \langle
Y^n\rangle_t $, that $X^n$ coincides with $X$ on
the interval $[0,\gamma_n]$ and that the new coefficients satisfy
assumption $(H'_1)$. Consequently,
\begin{eqnarray*}
\lim_{\varepsilon\downarrow0} \varepsilon^{-2} E \biggl[\int
_0^{\gamma_n}(X_t - X_{\eta(t)})^2
A_t \,dt \biggr]&=&\lim_{\varepsilon\downarrow0} \varepsilon^{-2}
E \biggl[ \biggl(\int_0^{T}
\bigl(X^n_t - X^n_{\eta^n(t)}
\bigr)^2 \,dY^n_t \biggr)^2 \biggr]
\nonumber
\\
&=& E \biggl[\int_0^{T} {A^n_t}
\frac{f(\underline a^n_t, \overline a^n_t)}{g(\underline a^n_t,
\overline
a^n_t)}\,dt \biggr] \\
&= &E \biggl[\int_0^{\gamma_n}
{A_t} \frac{f(\underline a_t, \overline a_t)}{g(\underline a_t, \overline
a_t)} \,dt \biggr],%\label{doublelim}
\end{eqnarray*}
which implies, by assumption $(\mathit{HA})$, that
\[
E \biggl[\int_0^{\gamma_n} {A_t}
\frac{f(\underline a_t, \overline a_t)}{g(\underline a_t, \overline
a_t)} \,dt \biggr] \leq E \biggl[\sup_{0\leq s \leq T}\max(
\underline a_s,\overline a_s)^2\int
_0^{T} A_t \,dt \biggr]<+\infty,
\]
and so by Fatou's lemma,
\[
E \biggl[\int_0^{T} {A_t}
\frac{f(\underline a_t, \overline a_t)}{g(\underline a_t, \overline
a_t)} \,dt \biggr] \leq E \biggl[\sup_{0\leq s \leq T}\max(
\underline a_s,\overline a_s)^2\int
_0^{T} A_t \,dt \biggr]<+\infty.
\]
Therefore, by dominated convergence
\[
\lim_n E \biggl[\int_{\gamma_n}^T
{A_t} \frac{f(\underline a_t, \overline a_t)}{g(\underline a_t, \overline
a_t)} \,dt \biggr] = 0.
\]
On the other hand,
\[
\varepsilon^{-2} E\int_{\gamma_n}^T
(X_t - X_{\eta(t)})^2 A_t \,dt \leq E
\biggl[\sup_{0\leq s \leq T}\max(\underline a_s,\overline
a_s)^2\int_{\gamma_n}^T
A_t \,dt \biggr].
\]
The right-hand side does not depend on $\varepsilon$ and converges to
zero as $n\to\infty$ by the
dominated convergence theorem. Therefore, the left-hand side can be made
arbitrarily small independently of $\varepsilon$, and the result follows.
\end{pf}

\textit{Step} 2. \textit{Change of probability measure}.
We first prove the following important lemma.
%
%le3 #&#
\begin{lemma}\label{supzero}
Under the assumption $H'_{1}$, almost surely,
\[
\lim_{\varepsilon\to0} \sup_{i: T^\varepsilon_i \leq T }
\bigl(T^\varepsilon_{i+1}-T^\varepsilon_i\bigr) =
0.
\]
\end{lemma}
\begin{pf}
In this proof, let us fix $\omega\in\Omega$. By way of contradiction,
assume that there exists a constant $C>0$ and a sequence
$\{\varepsilon_n\}_{n\geq0}$ converging to zero such that for every
$n$, there exists $i(n)$ with $T^{\varepsilon_n}_{i(n)+1} -
T^{\varepsilon_n}_{i(n)} > C$. From the sequences
$\{T^{\varepsilon_n}_{i(n)+1}\}_{n}$ and
$\{T^{\varepsilon_n}_{i(n)}\}_n$ we can extract two subsequences
$\{T^{\varepsilon_{\phi(n)}}_{i(\phi(n))+1}\}_{n}$ and
$\{T^{\varepsilon_{\phi(n)}}_{i(\phi(n))}\}_n$ converging to some
limiting values $T_1< T_2$. For $n$ big enough, there exists a nonempty
interval $\mathcal I$ which is a subset of both $(T_1,T_2)$ and
$(T^{\varepsilon_{\phi(n)}}_{i(\phi(n))+1},T^{\varepsilon_{\phi
(n)}}_{i(\phi(n))})$. Now
using that $\sup_{t,s \in
(T^{\varepsilon_{\phi(n)}}_{i(\phi(n))+1},T^{\varepsilon_{\phi
(n)}}_{i(\phi(n))})}
|X_t-X_s| \leq2B\varepsilon_{\phi(n)}$, we obtain that $\sup_{s,t\in
\mathcal I} |X_t-X_s|=0$, which cannot hold since $X$ is an
infinite activity process.
\end{pf}

Let $\Delta T_{i+1}=T_{i+1}\wedge T - T_i\wedge T$. The
goal of this step is to show that
%
%e32 #&#
%
\begin{eqnarray}\label{lim0}
&&\lim_{\varepsilon\downarrow0} \varepsilon^{-2} E \biggl[\int
_0^{T}(X_t - X_{\eta(t)})^2
A_t \,dt \biggr]
\nonumber
\\[-8pt]
\\[-8pt]
\nonumber
&&\qquad= \lim_{\varepsilon\downarrow0} E^Q
\Biggl[ \sum_{i=1}^\infty
Z^{-1}_{T_i\wedge T} A_{T_i\wedge T}\varepsilon^{-2} \int
_{T_i\wedge T}^{T_{i+1}\wedge T}(X_t - X_{T_i})^2
\,dt \Biggr].
\end{eqnarray}
We have
\begin{eqnarray*}
&&\varepsilon^{-2}E \biggl[\int_0^{T}(X_t
- X_{\eta(t)})^2 A_t \,dt \biggr]\\
&&\qquad=
\varepsilon^{-2}\sum_{i=0}^{+\infty}E
\biggl[\int_{T_i\wedge
T}^{T_{i+1}\wedge T}(X_t -
X_{T_i})^2 (A_t-A_{T_i}) \,dt \biggr]
\\
&&\qquad\quad{}+\varepsilon^{-2}\sum_{i=0}^{+\infty}E^Q
\biggl[Z^{-1}_{T_{i+1}\wedge
T}A_{T_i}\int_{T_i\wedge T}^{T_{i+1}\wedge T}(X_t
- X_{T_i})^2\,dt \biggr].
\end{eqnarray*}
Since for $t\in[T_i,T_{i+1})$, $(X_t - X_{T_i})^2\leq B^2\varepsilon
^{2}$, using the boundedness of $A$, \eqref{lim0}~will follow, provided
we show that
%
%e33 #&#
%
\begin{equation}
\label{lim1} \lim_{\varepsilon\downarrow0}\sum_{i=0}^{+\infty}E
\biggl[\int_{T_i\wedge T}^{T_{i+1}\wedge T}|A_t-A_{T_i}|\,dt
\biggr]=0
\end{equation}
and
%
%e34 #&#
%
\begin{equation}
\label{lim2}\lim_{\varepsilon\downarrow0}\sum_{i=0}^{+\infty}E^Q
\bigl[\bigl|Z^{-1}_{T_{i+1}\wedge T}-Z^{-1}_{T_{i}\wedge
T}\bigr|\Delta
T_{i+1} \bigr]=0.
\end{equation}
Limit \eqref{lim1} follows from the dominated convergence theorem ($A$
is bounded by assumption $(H'_1)$ and $A_{\eta(t)}\to A_t$ almost
everywhere on
$[0,T]$ since
$A$ is c\`adl\`ag and by Lemma \ref{supzero}).
Using the fact that $Z^{-1}$ has finite quadratic variation together
with Lemma \ref{supzero} and the Cauchy--Schwarz inequality, we get that,
in probability,
\[
\lim_{\varepsilon\downarrow0}\sum_{i=0}^{+\infty
}\bigl|Z^{-1}_{T_{i+1}\wedge T}-Z^{-1}_{T_{i}\wedge T}\bigr|
\Delta T_{i+1}=0.
\]
Then \eqref{lim2} follows from the integrability of
$\mbox{sup}_{t\in[0,T]}Z_t^{-1}$, which is part of assumption
$(H'_1)$.

\textit{Step} 3.
First, observe that by the dominated convergence theorem, since $\sup_i
\Delta T_i$ tends to zero, \eqref{lim0} is equal to
\begin{eqnarray}
S_1:=\lim_{\varepsilon\downarrow0} S_1^\varepsilon
\nonumber\\
\eqntext{\displaystyle\mbox {with } S_1^\varepsilon:=E^Q \Biggl[ \sum
_{i=0}^\infty 1_{T_i\leq T}
A_{T_i}Z^{-1}_{T_i}{\varepsilon^{-2}
E^{Q}_{\mathcal F_{T_i}} \biggl[\int_{T_i}^{T_{i+1}}(X_t
- X_{T_i})^2 \,dt \biggr]} \Biggr].}
\end{eqnarray}
For this expression to be well defined we extend the processes
$\lambda$, $b$, $\underline a$, $\overline a$ by arbitrary
constant values beyond $T$ and define the process $X$ for $t\geq
T$ accordingly.

Define a family of continuous increasing processes
$(\Lambda_s(t))_{t\geq0}$ indexed by $s\geq0$ by $\Lambda_s(t) =
\int_s^{s+t} \lambda_r \,dr$, the family of filtrations $\mathcal
G^i_t = \mathcal F_{T_i + t}$ and of processes $(\tilde
X^i_t)_{t\geq0}$ and $(\hat X^i_t)_{t\geq0}$ by
\[
\hat X^i_t = X_{T_i + \Lambda_{T_i}^{-1}(t)} - X_{T_i} -
\int_{T_i}^{T_i + \Lambda_{T_i}^{-1}(t)} b_s \,ds,\qquad \tilde
X^i_t = X_{T_i + \Lambda_{T_i}^{-1}(t)} - X_{T_i}.
\]
The process $(\hat X^i_t)_{t\geq0}$ is a $(G^i_t)$-semimartingale
with (deterministic) characteristics $(0,\nu,0)$ under $Q$, and
therefore, it is a $(G^i_t)$-L\'evy process under $Q$ (Theorem~II.4.15 in \cite{jacodshiryaev}).

Let $\tilde\tau_i = \inf\{t\geq0\dvtx \tilde X^i_{t} \notin
[-\underline a_{T_i} \varepsilon, \overline a_{T_i}
\varepsilon]\}$. Using a change of variable formula we obtain that
\[
\int_{T_i}^{T_{i+1}}(X_t -
X_{T_i})^2\,dt = \int_0^{\tilde\tau_i}
\frac{(\tilde X_s^i)^2}{\lambda(T_i + \Lambda^{-1}_{T_i}(s))} \,ds.
\]
Using the c\`adl\`ag property of $\lambda$ together with the various
boundedness assumptions and the integrability of $\sup_{0\leq t \leq
T} Z_t^{-1}$, we easily get that
\[
S_1 = \lim_{\varepsilon\downarrow0} E^Q \Biggl[ \sum
_{i=0}^\infty 1_{T_i\leq T}
\frac{A_{T_i}Z^{-1}_{T_i}}{\lambda_{T_i}}{\varepsilon^{-2} E^{Q}_{\mathcal F_{T_i}}
\biggl[\int_{0}^{\tilde\tau_{i}} \bigl(\tilde
X_t^i\bigr)^2 \,dt \biggr]} \Biggr].
\]
Then we obviously have that
\[
S_1 = \lim_{\varepsilon\downarrow0} E^Q \Biggl[ \sum
_{i=0}^\infty 1_{T_i\leq T}
\frac{A_{T_i}Z^{-1}_{T_i}}{\lambda_{T_i}} \frac
{T_{i+1}-T_i}{E^Q_{\mathcal F_{T_i}}[T_{i+1}-T_i]}{\varepsilon^{-2}
E^{Q}_{\mathcal F_{T_i}} \biggl[\int_{0}^{\tilde\tau_{i}}
\bigl(\tilde X_t^i\bigr)^2 \,dt \biggr]}
\Biggr].
\]
Now note that
%
%e35 #&#
%
\begin{equation}
T_{i+1}-T_i = \int_0^{\tilde\tau_i}
\frac{ds}{\lambda(T_i +
\Lambda^{-1}_{T_i}(s))}. \label{deltaT}
\end{equation}
Then
\begin{eqnarray*}
&&E^Q \Biggl[ \sum_{i=0}^\infty
1_{T_i\leq T} \frac{A_{T_i}Z^{-1}_{T_i}}{\lambda_{T_i}} \frac
{T_{i+1}-T_i}{E^Q_{\mathcal F_{T_i}}[T_{i+1}-T_i]}{\varepsilon^{-2}
E^{Q}_{\mathcal F_{T_i}} \biggl[\int_{0}^{\tilde\tau_{i}}
\bigl(\tilde X_t^i\bigr)^2 \,dt \biggr]}
\Biggr]
\\
&&\qquad = E^Q \Biggl[ \sum_{i=0}^\infty
1_{T_i\leq T} {A_{T_i}Z^{-1}_{T_i}}
\frac{T_{i+1}-T_i}{E^Q_{\mathcal
F_{T_i}}[\tilde\tau_i]}{\varepsilon^{-2} E^{Q}_{\mathcal F_{T_i}}
\biggl[\int_{0}^{\tilde\tau_{i}} \bigl(\tilde
X_t^i\bigr)^2 \,dt \biggr]} \Biggr] +
R^\varepsilon
\end{eqnarray*}
with
\begin{eqnarray*}
&&\bigl|R^\varepsilon\bigr| \leq C E^Q \Biggl[\sum
_{i=0}^\infty1_{T_i\leq T} Z^{-1}_{T_i}
(T_{i+1}-T_i) \\
&&\hspace*{77pt}{}\times\biggl\llvert \frac{\lambda_{T_i}^{-1} E_{\mathcal F_{T_i}}[
\tilde\tau_i] - E_{\mathcal F_{T_i}}
[\int_0^{\tilde\tau_i} {ds}/{(\lambda(T_i +
\Lambda^{-1}_{T_i} (s)))}]}{E_{\mathcal F_{T_i}}
[\int_0^{\tilde\tau_i} {ds}/{(\lambda(T_i +
\Lambda^{-1}_{T_i} (s)))}]}\biggr
\rrvert \Biggr].
\end{eqnarray*}
Using \eqref{deltaT}, we obtain that
\begin{eqnarray*}
\bigl|R^\varepsilon\bigr| &\leq& C E^Q \Biggl[\sum
_{i=0}^\infty1_{T_i\leq T} Z^{-1}_{T_i}
\biggl\llvert {\lambda_{T_i}^{-1} E_{\mathcal F_{T_i}}[ \tilde
\tau_i] - E_{\mathcal F_{T_i}} \biggl[\int_0^{\tilde\tau_i}
\frac{ds}{\lambda(T_i +
\Lambda^{-1}_{T_i} (s))}\biggr]}\biggr\rrvert \Biggr]
\\
&\leq& C E^Q \Biggl[\sum_{i=0}^\infty1_{T_i\leq T}
Z^{-1}_{T_i} \int_0^{\tilde\tau_i}
\biggl\llvert {\frac{1}{\lambda_{T_i}} - \frac{1}{\lambda(T_i +
\Lambda^{-1}_{T_i} (s))}}\biggr\rrvert \,ds \Biggr]
\\
& \leq& C E^Q \Biggl[\sum_{i=0}^\infty1_{T_i\leq T}
Z^{-1}_{T_i} \int_{T_i}^{T_{i+1}}
\biggl\llvert {\frac{1}{\lambda_{T_i}} - \frac{1}{\lambda(s)}}\biggr\rrvert \,ds \Biggr],
\end{eqnarray*}\eject\noindent
which is easily shown to converge to zero. Consequently, we conclude that
%
%e36 #&#
%
\begin{equation}
S_1 = \lim_{\varepsilon\downarrow0} E^Q \Biggl[ \sum
_{i=0}^\infty 1_{T_i\leq T}
{A_{T_i}Z^{-1}_{T_i}} \frac{T_{i+1}-T_i}{E^Q_{\mathcal
F_{T_i}}[\tilde\tau_i]}{
\varepsilon^{-2} E^{Q}_{\mathcal F_{T_i}} \biggl[\int
_{0}^{\tilde\tau_{i}} \bigl(\tilde X_t^i
\bigr)^2 \,dt \biggr]} \Biggr].\label{S1tilde}
\end{equation}

\textit{Step} 4. \textit{Comparison of hitting times and associated
integrals}. We start with
the following lemma:
%
%le4 #&#
\begin{lemma}
Let $\kappa\in \mathbb R_+$ and $n \in\mathbb N$. Then
\[
\underline f^{\kappa,n}_\varepsilon(\underline a_{T_i},
\overline a_{T_i}) \leq E^{Q}_{\mathcal F_{T_i}} \biggl[ \biggl(
\int_0^{\tilde\tau_i} \bigl|\hat X_t^i\bigr|^\kappa
\,dt \biggr)^n \biggr] \leq \overline f^{\kappa,n}_\varepsilon(
\underline a_{T_i}, \overline a_{T_i})
\]
whenever the expression in the middle is well defined, where
$\underline f_\varepsilon$ and $\overline f_\varepsilon$ are
deterministic functions defined by
\begin{eqnarray*}
\underline f^{\kappa,n}_\varepsilon(a,b)&=& E^{Q} \biggl[
\biggl(\int_0^{\hat
\tau_1} |\hat X_t|^\kappa
\,dt \biggr)^n \biggr]\quad\mbox{and}\\
 \overline f^{\kappa,n}_\varepsilon(a,b)&=&
E^{Q} \biggl[ \biggl(\int_0^{\hat
\tau_2\wedge
\hat\tau^j} |
\hat X_t|^\kappa \,dt \biggr)^n \biggr],
\end{eqnarray*}
with $\hat X_t=\hat X_t^0$ and
\begin{eqnarray*}
\hat\tau_1 &=& \inf\bigl\{t\dvtx \hat X_t \leq-a
\varepsilon+ t B^2 \mbox{ or } \hat X_t \geq b\varepsilon- t
B^2\bigr\},
\\
\hat\tau_2 &=& \inf\bigl\{t\dvtx \hat X_t \leq-a
\varepsilon- t B^2 \mbox{ or } \hat X_t \geq b\varepsilon+ t
B^2\bigr\},
\\
\hat\tau^j &=& \inf\bigl\{t\dvtx |\Delta\hat X_t|\geq
\varepsilon(a+b) \bigr\}.
\end{eqnarray*}
\end{lemma}
The proof follows from the fact that $|\tilde X_t^i - \hat
X_t^i| \leq t B^2$ and that $\hat X$ is a $\mathcal G^i_t$-L\'evy
process under $Q$, and that a jump of size greater than
$\varepsilon(a+b)$ immediately takes the process $\tilde X^i$ out of
the interval.

%le5 #&#
\begin{lemma}\label{flimit.lm}
%
%e37 #&#
%
\begin{equation}
\lim_{\varepsilon\downarrow0} \varepsilon^{-(\kappa+\alpha) n
}\underline
f^{\kappa,n}_\varepsilon(a,b) = \lim_{\varepsilon
\downarrow0}
\varepsilon^{-(\kappa+\alpha) n
}\overline f^{\kappa,n}_\varepsilon(a,b) =
f^{*,\kappa,n}(a,b)\label{flimit}
\end{equation}
uniformly on $(a,b) \in[a_1,a_2] \times[b_1,b_2]$ for all $0<a_1
\leq a_2 < \infty$ and $0<b_1 \leq b_2 < \infty$, with
\[
f^{*,\kappa,n}(a,b) = E \biggl[ \biggl(\int_0^{\tau^*}
\bigl|X^*_t\bigr|^\kappa \,dt \biggr)^n \biggr],
\]
where $X^*$ is a strictly $\alpha$-stable process with L\'evy
density
\[
\nu^*(x) = \frac{c_+ 1_{x>0} + c_- 1_{x<0}}{|x|^{1+\alpha}}
\]
and $\tau^* = \inf\{t\geq0\dvtx X^*_t \notin(-a, b)
\}$.
\end{lemma}
\begin{pf}
For $\varepsilon>0$, let us define
$X_t^{\varepsilon}=\varepsilon^{-1}\hat{X}_{\varepsilon^\alpha
t}$, $X_t^{\varepsilon,1}=X_t^{\varepsilon}-tB^2\varepsilon^{\alpha
-1}$, $X_t^{\varepsilon,2}=X_t^{\varepsilon}+tB^2\varepsilon^{\alpha
-1}$ and
\begin{eqnarray*}
\tau_1^{\varepsilon,1}&=&\inf\bigl\{t, X_t^{\varepsilon,1}
\leq -a\bigr\},\qquad \tau_1^{\varepsilon,2}=\inf\bigl\{t,
X_t^{\varepsilon,2}\geq b\bigr\},
\\
\tau_2^{\varepsilon,1}&=&\inf\bigl\{t, X_t^{\varepsilon,2}
\leq -a\bigr\},\qquad \tau_2^{\varepsilon,2}=\inf\bigl\{t,
X_t^{\varepsilon,1}\geq b\bigr\}
\\
\tau_3^{\varepsilon,1}&=&\inf\bigl\{t, X_t^{\varepsilon}
\leq -a\bigr\},\qquad \tau_3^{\varepsilon,2}=\inf\bigl\{t,
X_t^{\varepsilon}\geq b\bigr\}.
\end{eqnarray*}
We write
$\tau_i^{\varepsilon}=\tau_i^{\varepsilon,1}\wedge\tau_i^{\varepsilon,2}$
for $i=1,2,3$. Similarly, we define
$\tau^{j,\varepsilon}:=\inf\{t\dvtx \allowbreak |\Delta
X_t^{\varepsilon}|\geq(a+b)\}$. Observe that by a change of variable
in the integral,
\begin{eqnarray*}
\varepsilon^{-(\kappa+\alpha) n
}\underline f^{\kappa,n}_\varepsilon(a,b) &=&
E^{Q} \biggl[ \biggl(\int_0^{\tau^{\varepsilon}_1}
\bigl|X^{\varepsilon}_t\bigr|^\kappa \,dt \biggr)^n
\biggr],
\\
\varepsilon^{-(\kappa+\alpha) n
}\overline f^{\kappa,n}_\varepsilon(a,b) &=&
E^{Q} \biggl[ \biggl(\int_0^{\tau^{\varepsilon}_2 \wedge\tau^{j,\varepsilon}}
\bigl|X^{\varepsilon}_t\bigr|^\kappa \,dt \biggr)^n
\biggr].
\end{eqnarray*}

From Lemma \ref{weakcvg}, we have that $X_t^{\varepsilon}$
converges to $X^*_t$ in Skorohod topology. From Skorohod
representation theorem, there exists some probability space on which
are defined a process $Y^*$ and a family of processes
$Y^{\varepsilon}$ such that $Y^{\varepsilon}$ and
$X^{\varepsilon}$ have the same law, $Y^{*}$ and $X^{*}$ have
the same law and $Y^{\varepsilon}$ converges to $Y^{*}$ almost
surely, for the Skorohod topology.

This implies that $Y^{\varepsilon,1}$ and $Y^{\varepsilon,2}$ also
converge to $Y^*$ almost surely, where $Y_t^{\varepsilon
,1}=Y_t^{\varepsilon}-tB^2\varepsilon^{\alpha-1}$ and
$Y_t^{\varepsilon,2}=Y_t^{\varepsilon}+tB^2\varepsilon^{\alpha-1}$. Now
using that the application which to a function $f$ in the Skorohod
space associates its first hitting time of a constant barrier is
continuous at almost all $f$ which are sample paths of strictly stable
processes (see Proposition VI.2.11 in \cite{jacodshiryaev} and its use
in \cite{rosenbaum.tankov.10}), we obtain that
$\sigma_i^{\varepsilon}$ converges almost surely to
$\sigma^{*}$ for $i=1,2,3$, where $\sigma_i^{\varepsilon}$ and
$\sigma^{*}$ are defined through $Y^{\varepsilon,1}$,
$Y^{\varepsilon,2}, Y^*$ in the same way as
$\tau_i^{\varepsilon}$ and $\tau^{*}$ through
$X^{\varepsilon,1}$, $X^{\varepsilon,2}, X^*$. Moreover,
since $\sigma_3^{\varepsilon} \leq\sigma^{j,\varepsilon}$
for all $\varepsilon$, we also have that
$\sigma_2^{\varepsilon}\wedge\sigma^{j,\varepsilon} \to
\sigma^*$ almost surely.

Now remark that, almost surely, $Y_t^{\varepsilon}$
converges almost everywhere in $t$ to~$Y_t^{*}$; see Proposition
VI.2.3 in \cite{jacodshiryaev}. Therefore, since
$|Y_t^{\varepsilon}|\mathrm{1}_{t\leq\sigma_1^{\varepsilon}}\leq
\mbox{max}(a,b)$ and $|Y_t^{\varepsilon}|\mathrm{1}_{t\leq\sigma
^{j,\varepsilon}\wedge\sigma_2^{\varepsilon}}\leq
\mbox{max}(a,b) + B^2 t$, using the dominated convergence theorem, we
obtain that almost surely
\begin{eqnarray*}
\biggl(\int_0^{\sigma_1^{\varepsilon}}\bigl|Y_t^{\varepsilon}\bigr|^\kappa
\,dt \biggr)^n&\rightarrow& \biggl(\int_0^{\sigma^{*}}\bigl|Y_t^{*}\bigr|^\kappa
\,dt \biggr)^n\quad\mbox{and}\\
 \biggl(\int_0^{\sigma_2^{\varepsilon} \wedge\sigma
^{j,\varepsilon}}\bigl|Y_t^{\varepsilon}\bigr|^\kappa
\,dt \biggr)^n&\rightarrow&\biggl(\int_0^{\sigma^{*}}\bigl|Y_t^{*}\bigr|^\kappa
\,dt \biggr)^n.
\end{eqnarray*}
Finally, we deduce that
\begin{eqnarray*}
\biggl(\int_0^{\tau_1^{\varepsilon}}\bigl|X_t^{\varepsilon}\bigr|^\kappa
\,dt \biggr)^n&\rightarrow& \biggl(\int_0^{\tau^{*}}\bigl|X_t^{*}\bigr|^\kappa
\,dt \biggr)^n\quad\mbox{and}\\
\biggl(\int_0^{\tau_2^{\varepsilon}\wedge\tau
^{j,\varepsilon}}\bigl|X_t^{\varepsilon}\bigr|^\kappa
\,dt \biggr)^n&\rightarrow &\biggl(\int_0^{\tau^{*}}\bigl|X_t^{*}\bigr|^\kappa
\,dt \biggr)^n,
\end{eqnarray*}
in law.

Now note that $\tau^{j,\varepsilon}$ is the first jump time
of a L\'evy process with characteristic triplet given by
$(0,\varepsilon^{\alpha}\nu|_{ (-(a+b)\varepsilon,(a+b)\varepsilon
)^c},0)$. Using that this process is a compound Poisson process, we get
\[
P\bigl[\tau^{j,\varepsilon} > T\bigr] \leq\exp\bigl\{-T \varepsilon^\alpha
\nu\bigl(\bigl(-\infty,-(a+b)\varepsilon\bigr] \cup\bigl[(a+b)\varepsilon,\infty\bigr)
\bigr) \bigr\},
\]
which, {by property \eqref{hbnd}}, implies
that the family $(\tau^{j,\varepsilon})_{\varepsilon>0}$ has
uniformly bounded exponential moment. This implies that the families
\[
\biggl(\int_0^{\tau_2^{\varepsilon}\wedge\tau^{j,\varepsilon
}}\bigl|X_t^{\varepsilon}\bigr|^\kappa
\,dt \biggr)^n \quad\mbox{and}\quad \biggl(\int_0^{\tau_1^{\varepsilon
}}\bigl|X_t^{\varepsilon}\bigr|^\kappa
\,dt \biggr)^n = \biggl(\int_0^{\tau_1^{\varepsilon} \wedge\tau
^{j,\varepsilon}}\bigl|X_t^{\varepsilon}\bigr|^\kappa
\,dt \biggr)^n,
\]
parameterized by $\varepsilon$, are uniformly integrable, and therefore
the proof of the convergence in \eqref{flimit} is complete.

It remains to show that the convergence in \eqref{flimit} is
uniform in $(a,b)$ over compact sets excluding zero. To do this, first
observe that $f^{*,\kappa,n}(a,b)$ is continuous in $(a,b)$ on compact sets
excluding zero (this is shown using essentially the same arguments as
above: continuity of the exit times for stable processes plus uniform
integrability). Second, since both $\underline f^{\kappa,n}_{\varepsilon
}$ and
$\overline f^{\kappa,n}_\varepsilon$ are increasing in $a$ and $b$, a
multidimensional version of Dini's theorem can be used to conclude that
the convergence is indeed
uniform.
\end{pf}

\textit{Step} 5. First, let us show that
\[
S_1 = \lim_{\varepsilon\downarrow0} E^Q \Biggl[ \sum
_{i=0}^\infty 1_{T_i\leq T}
{A_{T_i}Z^{-1}_{T_i}} \frac{T_{i+1}-T_i}{E^Q_{\mathcal
F_{T_i}}[\tilde\tau_i]}{
\varepsilon^{-2} E^{Q}_{\mathcal F_{T_i}} \biggl[\int
_{0}^{\tilde\tau_{i}} \hat X^2_t \,dt
\biggr]} \Biggr].
\]
Indeed, the absolute value of the difference between the expressions
under the limit here and in \eqref{S1tilde} is
bounded from above by
%
%e38 #&#
%
\begin{eqnarray}
\label{tildehat} &&E^Q \Biggl[ \sum_{i=0}^\infty
1_{T_i\leq T} {A_{T_i}Z^{-1}_{T_i}}
\frac{T_{i+1}-T_i}{E^Q_{\mathcal
F_{T_i}}[\tilde\tau_i]}{\varepsilon^{-2} E^{Q}_{\mathcal F_{T_i}}
\biggl[\int_{0}^{\tilde\tau_{i}} \bigl|(\tilde X_t -
\hat X_t) (\tilde X_t + \hat X_t)\bigr| \,dt \biggr]}
\Biggr]
\nonumber
\\
&&\qquad\leq C E^Q \Biggl[ \sum_{i=0}^\infty
1_{T_i\leq T} {Z^{-1}_{T_i}} \frac{T_{i+1}-T_i}{E^Q_{\mathcal
F_{T_i}}[\tilde\tau_i]}{
\varepsilon^{-2} E^{Q}_{\mathcal F_{T_i}} \bigl[\tilde
\tau_i^3 + \tilde\tau_i^2
\varepsilon \bigr]} \Biggr]
\\
&&\qquad\leq C E^Q \Biggl[ \sum_{i=0}^\infty
1_{T_i\leq T} {Z^{-1}_{T_i}} (T_{i+1}-T_i)
\frac{\varepsilon
^{-2}\overline f^{0,3}_\varepsilon(\underline a_{T_i},\overline
a_{T_i}) + \varepsilon^{-1}\overline f^{0,2}_\varepsilon(\underline
a_{T_i},\overline a_{T_i})
}{\underline
f^{0,1}_\varepsilon(\underline a_{T_i},\overline a_{T_i})}{ } \Biggr],\nonumber\hspace*{-15pt}
\end{eqnarray}
where $C$ is a constant which does not depend on $\varepsilon$. Using
Lemma \ref{flimit.lm} and the fact that $\alpha>1$, we get
\[
\sup_{{1}/{B} \leq a,b \leq B}\frac{\varepsilon^{-2}\overline
f^{0,3}_\varepsilon(a,b) + \varepsilon^{-1}\overline f^{0,2}_\varepsilon(a,b)
}{\underline
f^{0,1}_\varepsilon(a,b)} \to0\qquad \mbox{as $\varepsilon
\to0$.}
\]
This, together with the fact that $E^Q[\sup_{0\leq t \leq T}
Z_t^{-1}]<\infty$, enables us to apply the dominated convergence
theorem and conclude that \eqref{tildehat} goes to zero.

Finally, we have that
\begin{eqnarray*}
S_1 &\leq&\limsup_{\varepsilon\downarrow0} E^Q \Biggl[
\sum_{i=0}^\infty 1_{T_i\leq T}
{A_{T_i}Z^{-1}_{T_i}} (T_{i+1}-T_i)
\frac{\varepsilon
^{-2-\alpha}\overline
f^{2,1}_\varepsilon(\underline a_{T_i},\overline a_{T_i})}{\varepsilon
^{-\alpha}\underline
f^{0,1}_\varepsilon(\underline a_{T_i},\overline
a_{T_i})}{} \Biggr],
\\
S_1 &\geq&\limsup_{\varepsilon\downarrow0} E^Q \Biggl[
\sum_{i=0}^\infty 1_{T_i\leq T}
{A_{T_i}Z^{-1}_{T_i}} (T_{i+1}-T_i)
\frac{\varepsilon
^{-2-\alpha}\underline f^{2,1}_\varepsilon(\underline a_{T_i},\overline
a_{T_i})}{\varepsilon^{-\alpha}\overline
f^{0,1}_\varepsilon(\underline a_{T_i},\overline a_{T_i})}{} \Biggr].
\end{eqnarray*}
Using for $(\kappa,n)=(0,1)$ and $(\kappa,n)=(2,1)$ the uniform
convergence on\break  $[1/B,B]$ of $\varepsilon^{-(\kappa+\alpha) n
}\underline f^{\kappa,n}_\varepsilon$\vspace*{2pt} and
$\varepsilon^{-(\kappa+\alpha) n
}\overline f^{\kappa,n}_\varepsilon$ toward $f^{*,\kappa,n}$ which is
continuous, together with a Riemann-sum type argument and the dominated
convergence theorem, we obtain that
\[
S_1 = E^Q \biggl[\int_0^T
A_t Z^{-1}_t \frac{
f^{*,2,1}(\underline a_t,\overline a_t)}{
f^{*,0,1}(\underline a_t,\overline a_t)} \,dt \biggr] = E
\biggl[\int_0^T A_t
\frac{
f^{*,2,1}(\underline a_t,\overline a_t)}{
f^{*,0,1}(\underline a_t,\overline a_t)} \,dt \biggr].
\]

%s5 #&#
\section{Preliminaries for the proof of Theorem \protect\texorpdfstring{\ref{cost.thm}}{2}}\label{prooflm}
In this section, we prove some technical lemmas concerning the uniform
integrability of the hitting time counts and the overshoots, which are
needed for the proof of Theorem~\ref{cost.thm}.

%le6 #&#
\begin{lemma}\label{overshoot.lm}
Under the {assumption $(\mathit{HX})$}, for all $\beta\in[0,\alpha)$ and
$\varepsilon>0$,
%
%e39 #&#
%
\begin{eqnarray}
\label{overbound}&& E_{\mathcal F_{T_i}} \bigl[|X_{T_{i+1}}-X_{T_i}|^\beta
\bigr]\nonumber\\
&&\qquad\leq c \varepsilon ^{\beta-1} \max\bigl\{\underline
a_{T_i}^{\beta-1},\overline a_{T_i}^{\beta-1}\bigr
\}E_{\mathcal F_{T_i}} \biggl[\int_{T_i}^{T_{i+1}}
|b_s|\,ds \biggr]
\nonumber
\\[-8pt]
\\[-8pt]
\nonumber
&&\qquad\quad{} + c \varepsilon^{\beta-\alpha} \max\{\underline a_{T_i},\overline
a_{T_i}\}^{\beta\vee(2-\alpha)}\\
&&\qquad\quad{}\times\min\{\underline a_{T_i},\overline
a_{T_i}\}^{(\beta-2)\wedge(-\alpha)} E_{\mathcal F_{T_i}} \biggl[\int
_{T_i}^{T_{i+1}} {\widehat K_s} \,ds \biggr],\nonumber
\end{eqnarray}
provided that the right-hand side has finite expectation.\vadjust{\eject}
\end{lemma}
%
%co1 #&#
\begin{corollary}\label{lower.cor}
{Under the assumption $(\mathit{HX})$}, for all $\varepsilon>0$,
\begin{eqnarray*}
\varepsilon^\alpha&\leq& c \varepsilon^{\alpha-1} \min\{\underline
a_{T_i},\overline a_{T_i}\}^{-1}E_{\mathcal F_{T_i}}
\biggl[\int_{T_i}^{T_{i+1}} |b_s|\,ds \biggr]
\\
&&{}+ c \min\{\underline a_{T_i},\overline a_{T_i}
\}^{-\alpha} E_{\mathcal F_{T_i}} \biggl[\int_{T_i}^{T_{i+1}}
{\widehat K_s} \,ds \biggr],
\end{eqnarray*}
provided that the right-hand side has finite expectation.
\end{corollary}
\begin{pf}%{Proof of Corollary}
Apply Lemma \ref{overshoot.lm} with
$\beta' = 0$,
$\underline a_{T_i}' = \overline a_{T_i}' = \min\{\underline
a_{T_i},\overline a_{T_i}\}$; then multiply both sides of
\eqref{overbound} by $\varepsilon^\alpha$ and use the fact that
the hitting time of the new barrier is smaller than $T_{i+1}$.
\end{pf}

\begin{pf*}{Proof of Lemma \ref{overshoot.lm}}
{First of all, from \eqref{hxbnd} we easily deduce
by integration by parts that
%
%e40 #&#
%
\begin{eqnarray}\label{upperK2}
\int_{x<|z|\leq1} |z| \mu(dt\times dz) &<& C \widehat K_t
x^{1-\alpha
} \quad\mbox{and}
\nonumber
\\[-8pt]
\\[-8pt]
\nonumber
 \int_{|z|\leq x} z^2
\mu(dt\times dz) &<& C\widehat K_t x^{2-\alpha}
\end{eqnarray}
for all $x>0$, for some constant $C<\infty$.}

For this proof, let
\begin{eqnarray*}
f(x)&:=& x^2 1_{0\leq x\leq2\overline a_{T_i} \varepsilon} (2\overline a_{T_i}
\varepsilon)^{\beta-2}+ |x|^\beta1_{x> 2\overline a_{T_i}
\varepsilon} +
x^2 1_{-2\underline a_{T_i} \varepsilon\leq x \leq0} (2\underline a_{T_i}
\varepsilon)^{\beta-2}\\
&&{}+ |x|^\beta1_{x<- 2\underline a_{T_i}
\varepsilon}.
\end{eqnarray*}
By It\^o's formula,
%
%e41 #&#
%
\begin{eqnarray}
\label{itoovershoot} &&2^{\beta-2}{E_{\mathcal F_{T_i}}
\bigl[|X_{T_{i+1}}-X_{T_i}|^\beta
\bigr]}\nonumber\\
&&\qquad\leq E_{\mathcal F_{T_i}}\bigl[f({X_{T_{i+1}}-X_{T_i}})\bigr]
\nonumber\\
&&\qquad= E_{\mathcal
F_{T_i}} \biggl[\int_{T_i}^{T_{i+1}}
f'({X_s - X_{T_i}}) b_s\,ds
\biggr]
\nonumber\\
&&\qquad\quad{} + E_{\mathcal
F_{T_i}} \biggl[\int_{T_i}^{T_{i+1}} \int
_{\mathbb R} \bigl\{ f({X_{s} + z - X_{T_i}})
- f({X_{s} - X_{T_i}})\\
&&\hspace*{158pt}{} - f'({X_{s}
- X_{T_i}}) z 1_{|z|\leq1} \bigr\} \mu(ds\times dz) \biggr]
\nonumber\\
&&\qquad\quad{} + E_{\mathcal
F_{T_i}} \biggl[\int_{T_i}^{T_{i+1}} \int
_{\mathbb R} \bigl\{ f({X_{s-} + z - X_{T_i}})
\nonumber\\
&&\hspace*{93pt}\qquad\quad{}- f({X_{s} - X_{T_i}}) \bigr\} (M-\mu) (ds\times dz)
\biggr].\nonumber
\end{eqnarray}
The first term in the right-hand side satisfies
\begin{eqnarray*}
&&E_{\mathcal
F_{T_i}} \biggl[\int_{T_i}^{T_{i+1}}
f'({X_s - X_{T_i}}) b_s\,ds
\biggr] \\
&&\qquad\leq(2\varepsilon)^{\beta-1}\max\bigl\{\underline
a_{T_i}^{\beta
-1},\overline a_{T_i}^{\beta-1}\bigr
\}E_{\mathcal
F_{T_i}} \biggl[\int_{T_i}^{T_{i+1}}
|b_s|\,ds \biggr].
\end{eqnarray*}
For the second term, we denote $A_s:= \{z \dvtx X_s + z - X_{T_i} \in
(-2\underline a_{T_i} \varepsilon, 2 \overline a_{T_i}
\varepsilon)\}$ and decompose it into two terms,
\begin{eqnarray*}
&&E_{\mathcal
F_{T_i}} \biggl[\int_{T_i}^{T_{i+1}} \int
_{A^c_s} \bigl\{ f({X_{s} + z - X_{T_i}})
\\
&&\hspace*{73pt}{}- f({X_{s} - X_{T_i}}) - f'({X_{s}
- X_{T_i}}) z 1_{|z|\leq1} \bigr\} \mu(ds\times dz) \biggr]
\\
&&\qquad\leq C E_{\mathcal
F_{T_i}} \biggl[\int_{T_i}^{T_{i+1}}
\int_{(-\underline a_{T_i} \varepsilon,\overline
a_{T_i} \varepsilon)^c} \bigl\{ |z|^\beta+ \varepsilon^{\beta-1}
\max\bigl\{\underline a_{T_i}^{\beta-1},\overline
a_{T_i}^{\beta-1}\bigr\}\\
&&\hspace*{237pt}{}\times|z| 1_{|z|\leq1} \bigr\} {\mu(ds
\times dz)} \biggr]
\\
&&\qquad\leq C \varepsilon^{\beta-\alpha} \bigl\{ \max\bigl\{\underline
a_{T_i}^{\beta-\alpha},\overline a_{T_i}^{\beta-\alpha}
\bigr\}+\max\bigl\{\underline a_{T_i}^{\beta-1},\overline
a_{T_i}^{\beta-1}\bigr\}\max\bigl\{\underline a_{T_i}^{1-\alpha},
\overline a_{T_i}^{1-\alpha}\bigr\} \bigr\} \\
&&\qquad\quad{}\times E_{\mathcal
F_{T_i}}
\biggl[\int_{T_i}^{T_{i+1}} { \widehat K_s}
\,ds \biggr]
\end{eqnarray*}
and
\begin{eqnarray*}
&& E_{\mathcal
F_{T_i}} \biggl[\int_{T_i}^{T_{i+1}} \int
_{A_s} \bigl\{ f({X_{s} + z - X_{T_i}})
- f({X_{s} - X_{T_i}}) \\
&&\hspace*{117pt}{}- f'({X_{s}
- X_{T_i}}) z 1_{|z|\leq1} \bigr\} \mu(ds\times dz) \biggr],
\end{eqnarray*}
which is smaller than
\begin{eqnarray*}
&&E_{\mathcal
F_{T_i}} \biggl[\int_{T_i}^{T_{i+1}} \int
_{A_s} \biggl\{\int_0^z
f''(X_s - X_{T_i} + x) (z-x)
\,dx \biggr\} \mu(ds\times dz) \biggr]
\\
&&\quad{} - E_{\mathcal
F_{T_i}} \biggl[\int_{T_i}^{T_{i+1}} \int
_{A_s} \bigl\{ f'({X_{s} -
X_{T_i}}) z 1_{|z|> 1} \bigr\} \mu(ds\times dz) \biggr]
\\
&&\qquad\leq C \varepsilon^{\beta-2}\max\bigl\{\underline a_{T_i}^{\beta-2},
\overline a_{T_i}^{\beta-2}\bigr\} E_{\mathcal
F_{T_i}} \biggl[\int
_{T_i}^{T_{i+1}} \int_{-3 \underline
a_{T_i}\varepsilon}^{3 \overline a_{T_i}\varepsilon}
z^2\mu(ds\times dz) \biggr]
\\
&&\qquad\quad{} + C\varepsilon^{\beta-1} \max\bigl\{\underline a_{T_i}^{\beta
-1},
\overline a_{T_i}^{\beta-1}\bigr\} E_{\mathcal
F_{T_i}} \biggl[\int
_{T_i}^{T_{i+1}} \int_{-3 \underline
a_{T_i}\varepsilon}^{3 \overline a_{T_i}\varepsilon}
z^2 \mu(ds\times dz) \biggr]
\\
&&\qquad\leq C \varepsilon^{\beta-\alpha} \bigl(\max\bigl\{\underline
a_{T_i}^{\beta
-2},\overline a_{T_i}^{\beta-2}
\bigr\} +\varepsilon\max\bigl\{\underline a_{T_i}^{\beta
-1},
\overline a_{T_i}^{\beta-1}\bigr\}\bigr)\\
&&\qquad\quad{}\times\max\bigl\{\underline
a_{T_i}^{2-\alpha},\overline a_{T_i}^{2-\alpha}
\bigr\} \\
&&\qquad\quad{}\times
E_{\mathcal
F_{T_i}} \biggl[\int_{T_i}^{T_{i+1}} {
\widehat K_s} \,ds \biggr],
\end{eqnarray*}
where we used \eqref{upperK2} in the last inequality.
Assembling the terms and doing some simple estimations yields the
statement of the lemma, provided we can show that the third term on
the right-hand side of \eqref{itoovershoot} is equal to zero. Splitting
it, once again, in two parts, we then get
\begin{eqnarray*}
&&E_{\mathcal
F_{T_i}} \biggl[\int_{T_i}^{T_{i+1}} \int
_{A_s^c}\bigl | f({X_{s-} + z - X_{T_i}}) -
f({X_{s} - X_{T_i}}) \bigr| \mu(ds\times dz) \biggr]
\\
&&\qquad \leq C E_{\mathcal
F_{T_i}} \biggl[\int_{T_i}^{T_{i+1}}
\int_{(-\infty,-\underline
a_{T_i}\varepsilon)\cup(\overline a_{T_i}\varepsilon,
\infty)}|z|^\beta\mu(ds\times dz) \biggr]\\
&&\qquad\leq C
\max\bigl\{\underline a_{T_i}^{\beta-\alpha},\overline
a_{T_i}^{\beta-\alpha}\bigr\} E_{\mathcal
F_{T_i}} \biggl[\int
_{T_i}^{T_{i+1}} {\widehat K_s} \,ds \biggr],
\end{eqnarray*}
and for the other term, using the ``isometry property'' of the
stochastic integral with respect
to the random measure together with \eqref{upperK2}, we obtain
\begin{eqnarray*}
&&E_{\mathcal
F_{T_i}} \biggl[ \biggl(\int_{T_i}^{T_{i+1}}
\int_{A_s}\frac{ f({X_{s-} + z -
X_{T_i}}) - f({X_{s} -
X_{T_i}})}{\max\{\underline
a_{T_i}^{\beta-2},\overline a_{T_i}^{\beta-2}\}} (M-\mu) (ds\times dz)
\biggr)^2 \biggr]
\\
&&\qquad\leq C\varepsilon^{2\beta-4} E_{\mathcal
F_{T_i}} \biggl[\int
_{T_i}^{T_{i+1}} \int_{-3\underline
a_{T_i}\varepsilon}^{3\overline a_{T_i}\varepsilon}
z^2\mu(ds\times dz) \biggr]\\
&&\qquad\leq C\varepsilon^{2\beta-2-\alpha}\max\bigl
\{\underline a_{T_i}^{2-\alpha},\overline a_{T_i}^{2-\alpha}
\bigr\} E_{\mathcal
F_{T_i}} \biggl[\int_{T_i}^{T_{i+1}} {
\widehat K_s} \,ds \biggr].
\end{eqnarray*}
Using the fact that both these terms have finite expectation by the
assumption of the lemma, we can now apply standard martingale
arguments to show that the third term in \eqref{itoovershoot} is equal
to zero.
\end{pf*}

%le7 #&#
\begin{lemma}\label{sumovershoot.lm}
Assume {$(\mathit{HX})$ and $(\mathit{HA}_2)$.} Let
$\{\tau_n\}$ be a sequence of stopping times converging to $T$ from
below. Then there exists $\varepsilon^*>0$ such that
%
%e42 #&#
%
\begin{equation}
\sup_{0<\varepsilon<\varepsilon^*} E \Biggl[ \Biggl(\varepsilon^{\alpha-\beta
}\sum
_{i=1}^{N^\varepsilon_T} |X_{T_{i}} -
X_{T_{i-1}}|^\beta \Biggr)^{1+\delta} \Biggr] <
\infty\label{sumovershoot1}
\end{equation}
and
%
%e43 #&#
%
\begin{equation}
\lim_{n\to\infty}\lim_{\varepsilon\downarrow0} E \Biggl[ \Biggl(
\varepsilon ^{\alpha-\beta}\sum_{i=N^\varepsilon_{\tau_n}+1}^{N^\varepsilon_T}
|X_{T_{i}} - X_{T_{i-1}}|^\beta \Biggr)^{1+\delta}
\Biggr] = 0.\label{sumovershoot2}
\end{equation}
\end{lemma}
\begin{pf}
In this proof, we shall use the notation
\begin{eqnarray*}
\overline\Lambda_t &= &\sup_{0\leq s\leq T}\bigl(\max\bigl\{
\underline a_s^{\beta
-1},\overline a_s^{\beta-1}
\bigr\}^{1+\delta}+\max\bigl\{\underline a_s^{(1+\delta)\beta
-1},
\overline a_s^{(1+\delta)\beta-1}\bigr\}\bigr) |b_t|^{1+\delta}\\
&&{}+
\sup_{0\leq s
\leq T}\max\{\underline a_s,\overline
a_s\}^{(\beta\vee(2-\alpha))(1+\delta)}\min\{\underline a_s,\overline
a_s\}^{((\beta-2)\wedge(-\alpha))(1+\delta)} {\widehat K_t}^{1+\delta}.
\end{eqnarray*}
%
%We extend the processes $\lambda$ and $b$
%by arbitrary constant values and the process $K$ by a deterministic
%function of $z$ satisfying \eqref{upperK} beyond $T$ and define the
%process $X$ for
%$t\geq T$ accordingly.
We now use a martingale {decomposition} of the sum of the increments.
So we write
\begin{eqnarray*}
\sum_{i=1}^{n} |X_{T_{i}} -
X_{T_{i-1}}|^\beta&=& M^1_n +
M^2_n + Z_n,
\\
M^1_n &=& \sum_{i=1}^{n}
\bigl\{|X_{T_{i}} - X_{T_{i-1}}|^\beta- E_{\mathcal F_{T_{i-1}}}
\bigl[|X_{T_{i}} - X_{T_{i-1}}|^\beta\bigr] \bigr\},
\\
 M^2_n& =& \sum_{i=1}^{n}
E_{\mathcal F_{T_{i-1}}}\bigl[|X_{T_{i}} - X_{T_{i-1}}|^\beta
\bigr] \biggl\{1- \frac{\int_{T_{i-1}}^{T_{i}} \Lambda^{T_{i-1}}_s
\,ds}{E_{\mathcal F_{T_{i-1}}} [\int_{T_{i-1}}^{T_{i}} \Lambda^{T_{i-1}}_s
\,ds ]} \biggr\},
\\
Z_n& =& \sum_{i=1}^{n}
E_{\mathcal F_{T_{i-1}}}\bigl[|X_{T_{i}} - X_{T_{i-1}}|^\beta
\bigr] \frac{\int_{T_{i-1}}^{T_{i}} \Lambda^{T_{i-1}}_s
\,ds}{E_{\mathcal F_{T_{i-1}}} [\int_{T_{i-1}}^{T_{i}} \Lambda^{T_{i-1}}_s
\,ds ]},
\end{eqnarray*}
where we write
\begin{eqnarray*}
\Lambda^{T_i}_s&:=& \varepsilon^{\alpha-1} \max\bigl
\{\underline a_{T_i}^{\beta-1},\overline a_{T_i}^{\beta-1}
\bigr\}|b_s|\\
&&{}+ \max\{\underline a_{T_i},\overline
a_{T_i}\}^{\beta\vee(2-\alpha)}\min\{\underline a_{T_i},\overline
a_{T_i}\}^{(\beta-2)\wedge(-\alpha)} {\widehat K_s}.
\end{eqnarray*}
The processes $M^1$ and $M^2$ are martingales with respect to the
discrete filtration $\mathcal F^d_n:= \mathcal
F_{T_{n}}$. Note that for every $\mathcal F$-stopping time $\tau\leq T$,
$N^\varepsilon_{\tau}$ is an $\mathcal
F^d$-stopping time.
The Burkholder inequality for a discrete-time martingale $M$ then writes
\begin{eqnarray*}
E\bigl[|M_{N^\varepsilon_T} - M_{N_\tau^\varepsilon}|^{1+\delta}\bigr]&\leq &C E
\Biggl[ \Biggl(\sum_{i=N^\varepsilon_{\tau}+1}^{N^\varepsilon_T}
(M_i-M_{i-1})^2 \Biggr)^{{(1+\delta)}/{2}}
\Biggr] \\
&\leq & C E \Biggl[\sum_{i=N^\varepsilon_{\tau}+1}^{N^\varepsilon_T}
|M_i-M_{i-1}|^{1+\delta} \Biggr],
\end{eqnarray*}
and therefore,
\begin{eqnarray*}
&&E\bigl[\bigl|\varepsilon^{\alpha-\beta}
\bigl(M^1_{N^\varepsilon_T}-M^1_{N^\varepsilon_\tau}
\bigr)\bigr|^{1+\delta}\bigr]\\
&&\qquad\leq C \varepsilon^{(\alpha-\beta)(1+\delta)}
E \Biggl[\sum
_{i=N^\varepsilon_\tau+1}^{N^\varepsilon_T} \bigl||X_{T_{i}} -
X_{T_{i-1}}|^\beta- E_{\mathcal F_{T_{i-1}}}\bigl[|X_{T_{i}} -
X_{T_{i-1}}|^\beta\bigr] \bigr|^{1+\delta} \Biggr]
\\
&&\qquad\leq C \varepsilon^{(\alpha-\beta)(1+\delta)} E \Biggl[\sum_{i=N^\varepsilon_\tau+1}^{N^\varepsilon_T}
E_{\mathcal
F_{T_{i-1}}}\bigl[|X_{T_{i}} - X_{T_{i-1}}|^{\beta(1+\delta)}
\bigr] \Biggr].
\end{eqnarray*}
By Lemma \ref{overshoot.lm}, this is smaller than
\begin{eqnarray*}
&&C E \biggl[\varepsilon^{\alpha(1+\delta)-1} \sup_{0\leq s\leq T}\max\bigl\{
\underline a_s^{\beta'-1},\overline a_s^{\beta'-1}
\bigr\}\int_{T_{N^\varepsilon_\tau}}^{{T_{N^\varepsilon_T }}} |b_s|\,ds
\\
&&\hspace*{22pt}{}+ \varepsilon^{\alpha\delta}\sup_{0\leq s\leq T} \max \{\underline
a_s,\overline a_s\}^{\beta'\vee(2-\alpha)} \min\{\underline
a_s,\overline a_s\}^{(\beta'-2)\wedge(-\alpha)} \int
_{T_{N^\varepsilon_\tau
}}^{{T_{N^\varepsilon_T}}} {\widehat K_s} \,ds \biggr]
\\
&&\qquad\leq C \varepsilon^{\alpha\delta} \biggl(E \biggl[ \int_{T_{N^\varepsilon_\tau}}^{{T_{N^\varepsilon_T}}}
\overline\Lambda_s \,ds \biggr] + E \biggl[ \int_{T_{N^\varepsilon_\tau
}}^{{T_{N^\varepsilon_T}}}
\overline\Lambda_s \,ds \biggr]^{{1}/{(1+\delta)}} \biggr),
\end{eqnarray*}
with $\beta' = \beta(1+\delta)$, where the last estimate can be
obtained, for example, by H\"older's inequality. %By the
%assumption of the corollary the above expression is bounded uniformly
%on $\varepsilon$ for $\varepsilon$ sufficiently small.

Similarly, the process $M^2$ satisfies
\begin{eqnarray*}
&&E\bigl[\bigl|\varepsilon^{\alpha-\beta}
 \bigl(M^2_{N^\varepsilon_T}-M^2_{N^\varepsilon
_\tau}
\bigr)\bigr|^{1+\delta}\bigr]\\
&&\qquad \leq C E \Biggl[\sum_{i=N^\varepsilon_\tau+ 1}^{N^\varepsilon_T}
\biggl\{\int_{T_{i-1}}^{T_{i}} \Lambda^{T_{i-1}}_s
\,ds - E_{\mathcal F_{T_{i-1}}} \biggl[\int_{T_{i-1}}^{T_{i}}
\Lambda^{T_{i-1}}_s \,ds \biggr] \biggr\}^{1+\delta} \Biggr]
\\
&&\qquad\leq C E \Biggl[\sum_{i=N^\varepsilon_\tau+ 1}^{N^\varepsilon_T} \biggl\{
\int_{T_{i-1}}^{T_{i}} \Lambda^{T_{i-1}}_s
\,ds \biggr\}^{1+\delta} \Biggr] \leq C E \biggl[\int_{T_{N^\varepsilon_\tau
}}^{T_{N^\varepsilon_T}}
\bigl(\Lambda^{\eta_s}_s\bigr)^{1+\delta} \,ds \biggr]
\\
&&\qquad\leq C E \biggl[\int_{T_{N^\varepsilon_\tau}}^{T_{N^\varepsilon
_T}} \overline
\Lambda_s \,ds \biggr].
\end{eqnarray*}
%
%which is shown to be uniformly bounded as above, using the condition
The process $Z$ can be treated along the same lines
as well, since by Lemma~\ref{overshoot.lm},
\[
E\bigl[\bigl|\varepsilon^{\alpha-\beta} (Z_{N^\varepsilon_T}-Z_{N^\varepsilon
_\tau})\bigr|^{1+\delta}
\bigr] \leq C E \biggl[ \biggl\{\int_{T_{N^\varepsilon_\tau}}^{T_{N^\varepsilon_T}}
\Lambda ^{\eta_s}_s \,ds \biggr\}^{1+\delta} \biggr] \leq C
E \biggl[\int_{T_{N^\varepsilon_\tau
}}^{T_{N^\varepsilon_T}} \overline
\Lambda_s \,ds \biggr].
\]
The three expressions above are uniformly bounded by the assumption of
the lemma, proving \eqref{sumovershoot1}. To show
\eqref{sumovershoot2}, observe that
\[
E \biggl[\int_{T_{N^\varepsilon_{\tau_n}}}^{T_{N^\varepsilon_T}} \overline
\Lambda_s \,ds \biggr] \leq E \biggl[\int_{\tau_n}^{T}
\overline\Lambda_s \,ds \biggr] + E \biggl[\sup_{i: T_i\leq T}
\int_{T_{i-1}}^{T_{i}} \overline \Lambda_s
\,ds \biggr].
\]
The first term does not depend on $\varepsilon$ and converges to zero
as $n\to\infty$ by the assumption of the lemma and the dominated
convergence. For the second term, we use Lemma \ref{supzero} and the
absolute continuity of the integral.
\end{pf}

In the case $\beta=0$, assumption {$(\mathit{HA}_2)$} can be
somewhat simplified.
%
%le8 #&#
\begin{lemma}
Assume {$(\mathit{HX})$ and $(\mathit{HA}'_2)$}. Let
$\{\tau_n\}$ be a sequence of stopping times converging to $T$ from
below.
Then there exists $\varepsilon^*>0$ such that
\[
\sup_{0< \varepsilon< \varepsilon^*} E\bigl[\bigl(\varepsilon^\alpha
N^\varepsilon_T\bigr)^{1+\delta}\bigr] <\infty
\]
and
\[
\lim_{n\to\infty}\lim_{\varepsilon\downarrow0} E \bigl[ \bigl(
\varepsilon ^{\alpha}\bigl(N^\varepsilon_T-N^\varepsilon_{\tau_n}
\bigr) \bigr)^{1+\delta} \bigr] = 0.
\]
\end{lemma}
\begin{pf}
We follow the proof of Lemma \ref{sumovershoot.lm}, taking $\beta=0$ and
\[
\Lambda^{T_i}_s:= \varepsilon^{\alpha-1} \min\{
\underline a_{T_i},\overline a_{T_i}\}^{-1}|b_s|+
\min\{\underline a_{T_i},\overline a_{T_i}\}^{-\alpha}
{\widehat K_s}
\]
and using Corollary \ref{lower.cor} instead of Lemma \ref{overshoot.lm}.
\end{pf}

\section{Proof of Theorem \protect\texorpdfstring{\ref{cost.thm}}{2}}\label{proofthm2}
\textit{Step} 1. \textit{Reduction to the case of bounded coefficients}. As
before, we start with the localization procedure.
%
%le9 #&#
\begin{lemma}
Assume that \eqref{cost0.eq} holds under the assumptions $(\mathit{HY})$,
$(\mathit{HX})$ and $(H'_1)$ and \eqref{cost.eq} holds under the assumptions $(\mathit{HY})$,
$(\mathit{HX})$ and $(H'_\rho)$ for some $\rho>
\frac{\alpha}{\alpha-\beta}\vee2$. Then Theorem \ref{cost.thm} holds.
\end{lemma}
\begin{pf}
The arguments related to the localization of $Z$ are the same or very
similar to those in Lemma \ref{loc.lm}, and so they are omitted. We set
$u^0(a,b)=1$ for any $(a,b)$. With the same notation as in
the proof of this lemma, and using \eqref{sumovershoot2} in the first
equality we then get, for $0\leq\beta<\alpha$,
\begin{eqnarray*}
&&\lim_{\varepsilon\downarrow0} \varepsilon^{\alpha-\beta} E \Biggl[\sum
_{i=1}^{N^\varepsilon_T} |X_{T_i}-X_{T_{i-1}}|^\beta
\Biggr]\\
&&\qquad= \lim_{n\to\infty} \lim_{\varepsilon\downarrow0}
\varepsilon^{\alpha-\beta}E \Biggl[\sum_{i=1}^{N^\varepsilon_{\gamma_n}}
|X_{T_i} - X_{T_{i-1}}|^\beta \Biggr]
\\
&&\qquad = \lim_{n\to\infty} \lim_{\varepsilon\downarrow0}
\varepsilon^{\alpha-\beta}E \biggl[\sum_{i\geq1: T^n_i\leq\gamma_n }
\bigl|X^n_{T_i} - X^n_{T_{i-1}}\bigr|^\beta
\biggr]
\\
&&\qquad = \lim_{n\to\infty} E \biggl[\int_0^{\gamma_n}
\lambda_t \frac{u^\beta(\underline a_t, \overline a_t)}{g(\underline a_t,
\overline
a_t)}\,dt \biggr] \\
&&\qquad= E \biggl[\int
_0^{T} \lambda_t
\frac{u^\beta(\underline a_t, \overline a_t)}{g(\underline a_t,
\overline
a_t)}\,dt \biggr],
\end{eqnarray*}
where the assumptions of the lemma are used to pass from the second
to the third line.

{\textit{Step} 2. \textit{Change of probability measure}.} The goal of
this step
is to show that
%
%e44 #&#
%
\begin{eqnarray}\label{cost2}
S_2&:=&\lim_{\varepsilon\downarrow0} \varepsilon^{\alpha-\beta} E
\Biggl[\sum_{i=1}^{N^\varepsilon_T} |X_{T_i} -
X_{T_{i-1}}|^\beta \Biggr]
\nonumber
\\[-8pt]
\\[-8pt]
\nonumber
&=& \lim_{\varepsilon\downarrow0}
\varepsilon^{\alpha-\beta}E^Q \Biggl[\sum
_{i=1}^{\infty} 1_{T_{i-1}\leq T} Z^{-1}_{T_{i-1}}|X_{T_i}
- X_{T_{i-1}}|^\beta \Biggr].
\end{eqnarray}
For the right-hand side to be well defined we extend the processes
$\lambda$, $b$, $\underline a$, $\overline a$ by arbitrary
constant values beyond $T$ and define the process $X$ for $t\geq
T$ accordingly. The case $\beta=0$ being straightforward, we
assume that $\beta>0$.

To prove \eqref{cost2}, it is enough to show that
%
%e45 #&#
%
\begin{equation}
\lim_{\varepsilon\downarrow0} E^Q \Biggl[\varepsilon^{\alpha-\beta}
\sum_{i=1}^{\infty} 1_{T_{i}\leq T}
\bigl(Z^{-1}_{T_{i}} -Z^{-1}_{T_{i-1}}
\bigr)|X_{T_i} - X_{T_{i-1}}|^\beta \Biggr] =
0\label{cost21}
\end{equation}
and
%
%e46 #&#
%
\begin{equation}
\lim_{\varepsilon\downarrow0} \varepsilon^{\alpha-\beta} E^Q \bigl[
Z^{-1}_{T_{N^\varepsilon_T}}|X_{T_{N^\varepsilon_T+1}} - X_{T_{N^\varepsilon_T}}|^\beta
\bigr] =0.\label{cost22}
\end{equation}
The second term can be shown to converge to zero using Lemma
\ref{overshoot.lm}. For the first term, for $1<\kappa<
\frac{\alpha\rho}{\alpha+ \beta\rho}$, H\"older's inequality yields
\begin{eqnarray*}
&&E^Q \Biggl[ \Biggl(\varepsilon^{\alpha-\beta}\sum
_{i=1}^{\infty} 1_{T_{i}\leq T} \bigl(Z^{-1}_{T_{i}}
-Z^{-1}_{T_{i-1}}\bigr)|X_{T_i} - X_{T_{i-1}}|^\beta
\Biggr)^\kappa \Biggr]
\\
&&\qquad \leq E^Q \Bigl[\sup_{0\leq t\leq T} Z^{-\rho}_t
\Bigr]^{{\kappa
}/{\rho}}\\
&&\qquad\quad{}\times E^Q \Biggl[ \Biggl(\varepsilon^{\alpha-\beta}
\sum_{i=1}^{\infty} 1_{T_{i}\leq T}
|X_{T_i} - X_{T_{i-1}}|^\beta \Biggr)^{\kappa\rho/(\rho-\kappa)}
\Biggr]^{{(\rho-\kappa)}/{\rho}},
\end{eqnarray*}
which is bounded by a constant for $\varepsilon$ sufficiently small by
Lemma \ref{sumovershoot.lm} (applied under~$Q$) (the assumptions are
satisfied because we are working under $H'_\rho$ and therefore all
coefficients are bounded). Therefore, the expression under the
expectation in
\eqref{cost21} is uniformly integrable under $Q$ as $\varepsilon
\downarrow
0$. On the other hand, by the Cauchy--Schwarz inequality,
\begin{eqnarray*}
&&\varepsilon^{\alpha-\beta}\sum_{i=1}^{\infty}
1_{T_{i}\leq T} \bigl|Z^{-1}_{T_{i}} -Z^{-1}_{T_{i-1}}\bigr||X_{T_i}
- X_{T_{i-1}}|^\beta
\\
&&\qquad\leq\varepsilon^{{(\alpha-\beta)}/{2}} \Biggl(\sum_{i=1}^{N^\varepsilon_T}
\bigl(Z_{T_i}^{-1} - Z_{T_{i-1}}^{-1}
\bigr)^2 \Biggr)^{{1}/{2}}\\
&&\qquad\quad{}\times \sup_{0\leq t\leq T}|X_t|^{\beta
/2}
\Biggl(\varepsilon^{\alpha-\beta}\sum_{i=1}^{N^\varepsilon_T}
|X_{T_i} - X_{T_{i-1}}|^\beta \Biggr)^{{1}/{2}}.
\end{eqnarray*}
Since $Z^{-1}$ has finite quadratic variation, and the last factor
is uniformly integrable under $Q$ by Lemma \ref{sumovershoot.lm},
due to the first deterministic factor, the whole expression
converges to zero in probability, and \eqref{cost21} follows.

\emph{Step 3.} Using the same notation as in the proof of
Theorem \ref{err.thm} (step 3), we have
\begin{eqnarray*}
&&\varepsilon^{\alpha-\beta}E^Q \Biggl[\sum
_{i=1}^{\infty} 1_{T_{i-1}\leq
T} Z^{-1}_{T_{i-1}}(T_i-T_{i-1})
\frac{E^Q_{\mathcal
F_{T_{i-1}}}|\tilde X_{\tilde\tau_i}|^\beta}{E^Q_{\mathcal
F_{T_{i-1}}}[T_i -
T_{i-1}]} \Biggr]
\\
&&\qquad= \varepsilon^{\alpha-\beta}E^Q \Biggl[\sum
_{i=1}^{\infty} 1_{T_{i-1}\leq
T} Z^{-1}_{T_{i-1}}
\lambda_{T_{i-1}}(T_i-T_{i-1})\frac{E^Q_{\mathcal
F_{T_{i-1}}}|\tilde X_{\tilde\tau_i}|^\beta}{E^Q_{\mathcal
F_{T_{i-1}}}[\tilde\tau_i]}
\Biggr] + R^\varepsilon,
\end{eqnarray*}
where one can show, using first Lemma \ref{overshoot.lm} and then
exactly the same arguments as in the proof of Theorem \ref{err.thm},
that $R^\varepsilon\to0$ as $\varepsilon\downarrow0$. Then, from
the previous step,
\begin{eqnarray*}
S_2& =& \lim_{\varepsilon\downarrow0} \varepsilon^{\alpha-\beta}E^Q
\Biggl[\sum_{i=1}^{\infty} 1_{T_{i-1}\leq
T}
Z^{-1}_{T_{i-1}}(T_i-T_{i-1})
\frac{E^Q_{\mathcal
F_{T_{i-1}}}|X_{T_i} -
X_{T_{i-1}}|^\beta}{E^Q_{\mathcal F_{T_{i-1}}}[T_i -
T_{i-1}]} \Biggr]
\\
& = &\lim_{\varepsilon\downarrow0} \varepsilon^{\alpha-\beta}E^Q
\Biggl[\sum_{i=1}^{\infty} 1_{T_{i-1}\leq
T}
Z^{-1}_{T_{i-1}}\lambda_{T_{i-1}}(T_i-T_{i-1})
\frac{E^Q_{\mathcal
F_{T_{i-1}}}|\tilde X_{\tilde\tau_i}|^\beta}{E^Q_{\mathcal
F_{T_{i-1}}}[\tilde\tau_i]} \Biggr].
\end{eqnarray*}
Our next goal is to replace $\tilde X_{\tilde\tau_i}$ with $\hat
X_{\hat\tau_i}$ in the above expression, where $\hat\tau_i =
\inf\{t\geq0\dvtx \hat X_t \notin[-\underline a_{T_i} \varepsilon,
\overline a_{T_i}\varepsilon]\}$. Let $a = \min(\underline a_{T_i},
\overline a_{T_i})$ and define
\[
f(x) = (\varepsilon a)^{\beta}\frac{(\beta- \varepsilon a)
({x}/{(\varepsilon a)} )^2 + 2-\beta}{2-\varepsilon a} 1_{|x|<\varepsilon a} +
|x|^\beta1_{|x|>\varepsilon a}.
\]
$f$ is a twice differentiable function satisfying for small enough
$\varepsilon$
%
%e47 #&#
%
\begin{equation}
\bigl|f'(x)\bigr|\leq C \varepsilon^{\beta-1}\quad \mbox{and}\quad
\bigl|f^{\prime
\prime}(x)\bigr|\leq C \varepsilon^{\beta-2} \label{fbnd},
\end{equation}
and hence
It\^o's formula can be applied. Then,
\begin{eqnarray*}
&&\bigl| E^Q_{\mathcal F_{T_{i-1}}}\bigl[|\tilde X_{\tilde\tau_i}|^\beta
-|\hat X_{\hat\tau_i}|^\beta\bigr] \bigr|
\\
&&\qquad\leq\bigl | E^Q_{\mathcal
F_{T_{i-1}}}
\bigl[f(\tilde X_{\tilde\tau_i}) -f(\hat X_{\tilde\tau_i})\bigr] \bigr|+ \bigl|
E^Q_{\mathcal
F_{T_{i-1}}}\bigl[f(\hat X_{\tilde\tau_i}) -f(\hat
X_{\hat\tau_i})\bigr]\bigr |.
\end{eqnarray*}
By definition of $\tilde X$ and $\hat X$ and because all coefficients
are bounded, the first term satisfies
\[
\bigl| E^Q_{\mathcal F_{T_{i-1}}}\bigl[f(\tilde X_{\tilde\tau_i}) -f(\hat
X_{\tilde\tau_i})\bigr] \bigr| \leq C \varepsilon^{\beta-1} E^Q_{\mathcal F_{T_{i-1}}}[
\tilde\tau_i].
\]
For the second term, we use It\^o's formula,
\begin{eqnarray*}
&&E^Q_{\mathcal F_{T_{i-1}}}\bigl[f(\hat X_{\tilde\tau_i}) -f(\hat
X_{\hat\tau_i})\bigr]\\
&&\qquad= E^Q_{\mathcal F_{T_{i-1}}} \biggl[\int
_{\hat
\tau_i\wedge\tilde
\tau_i}^{\hat\tau_i \vee\tilde\tau_i}\int_{\mathbb R} \bigl\{f(
\hat X_s + z) - f(\hat X_s) - z1_{|z|\leq1}
f'(\hat X_s)\bigr\}\nu(dz) \,ds \biggr]
\\
&&\qquad\quad{} + E^Q_{\mathcal F_{T_{i-1}}} \biggl[\int_{\hat\tau_i\wedge\tilde
\tau_i}^{\hat\tau_i \vee\tilde\tau_i}
\int_{\mathbb R} \bigl\{f(\hat X_{s-} + z) - f(\hat
X_{s-})\bigr\}\bigl(\widehat M(ds \times dz) - \nu(dz) \,ds\bigr)
\biggr],
\end{eqnarray*}
where $\widehat M$ is the jump measure of $\hat X$. It follows by
standard arguments that the local
martingale term has zero expectation. To deal with the first term we
use the bounds in \eqref{fbnd} and decompose the integrand as follows:
\begin{eqnarray*}
&&\biggl|\int_{\mathbb R} \bigl\{f(\hat X_s + z) - f(\hat
X_s) - z1_{|z|\leq1} f'(\hat X_s)
\bigr\}\nu(dz)\biggr |
\\
&&\qquad\leq C\varepsilon^{\beta-2}\int_{|z|\leq
\varepsilon} z^2
\nu(dz) + C\varepsilon^{\beta-1}\int_{|z|>
\varepsilon} |z| \nu(dz)
\leq C \varepsilon^{\beta-\alpha},
\end{eqnarray*}
so that finally
\[
\bigl| E^Q_{\mathcal F_{T_{i-1}}}\bigl[|\tilde X_{\tilde\tau_i}|^\beta
-|\hat X_{\hat\tau_i}|^\beta\bigr] \bigr| \leq C \varepsilon^{\beta-1}
E^Q_{\mathcal F_{T_{i-1}}}[\tilde\tau_i] + C
\varepsilon^{\beta-\alpha}E^Q_{\mathcal F_{T_{i-1}}}\bigl[|\tilde
\tau_i-\hat\tau_i|\bigr].
\]
Substituting this estimate into the formula for $S_2$, we then get
\[
S_2 = \lim_{\varepsilon\downarrow0} \varepsilon^{\alpha-\beta}E^Q
\Biggl[\sum_{i=1}^{\infty} 1_{T_{i-1}\leq
T}
Z^{-1}_{T_{i-1}}\lambda_{T_{i-1}}(T_i-T_{i-1})
\frac{E^Q_{\mathcal
F_{T_{i-1}}}|\hat X_{\hat\tau_i}|^\beta}{E^Q_{\mathcal
F_{T_{i-1}}}[\tilde\tau_i]} \Biggr] + \lim_{\varepsilon\downarrow0} R^\varepsilon
\]
with
\begin{eqnarray*}
\bigl|R^\varepsilon\bigr| &\leq& C\varepsilon^{\alpha-1}E^Q \Biggl[
\sum_{i=1}^{\infty} 1_{T_{i-1}\leq
T}
Z^{-1}_{T_{i-1}}\lambda_{T_{i-1}}(T_i-T_{i-1})
\Biggr]
\\
&&{}+ CE^Q \Biggl[\sum_{i=1}^{\infty}
1_{T_{i-1}\leq
T} Z^{-1}_{T_{i-1}}\lambda_{T_{i-1}}(T_i-T_{i-1})
\frac{E^Q_{\mathcal
F_{T_{i-1}}}|\tilde\tau_i - \hat\tau_i|}{E^Q_{\mathcal
F_{T_{i-1}}}[\tilde\tau_i]} \Biggr].
\end{eqnarray*}
The first expectation is bounded (because $\lambda$ is bounded) and
$Z^{-1}$ is integrable, and therefore the first term converges to
zero. For the second term, we observe (using the notation of the proof
of Theorem \ref{err.thm}, step 4) that
\[
\underline f^{0,1}_\varepsilon(\underline a_{T_i},
\overline a_{T_i})\leq E^Q_{\mathcal
F_{T_{i-1}}}[\tilde
\tau_i] \leq\overline f^{0,1}_\varepsilon(\underline
a_{T_i},\overline a_{T_i})
\]
and
\[
E^Q_{\mathcal
F_{T_{i-1}}}|\tilde\tau_i - \hat
\tau_i| \leq E^Q\bigl[\hat\tau_2 \wedge \hat
\tau^j - \hat\tau_1\bigr] \leq \overline
f^{0,1}_\varepsilon(\underline a_{T_i},\overline
a_{T_i}) - \underline f^{0,1}_\varepsilon(\underline
a_{T_i},\overline a_{T_i}).
\]
In view of Lemma \ref{flimit.lm} we then conclude that the second term
converges to zero as well. Finally, we have shown that
\[
S_2 = \lim_{\varepsilon\downarrow0} \varepsilon^{\alpha-\beta}E^Q
\Biggl[\sum_{i=1}^{\infty} 1_{T_{i-1}\leq
T}
Z^{-1}_{T_{i-1}}\lambda_{T_{i-1}}(T_i-T_{i-1})
\frac{u^\beta_\varepsilon
(\underline
a_{T_{i-1}}, \overline a_{T_{i-1}})}{E^Q_{\mathcal
F_{T_{i-1}}}[\tilde\tau_i]} \Biggr],
\]
where $u^\beta_\varepsilon$ is a deterministic function defined by
\[
u^\beta_\varepsilon(a,b) = E\bigl[|\hat X_{\hat\tau}|^\beta
\bigr], \qquad\hat \tau= \inf\bigl\{t\geq0\dvtx \hat X_t \notin(-a
\varepsilon,b\varepsilon)\bigr\}.
\]
Similar to the last step of the proof of Theorem \ref{err.thm},
we can now write
\begin{eqnarray*}
S_2 &\leq&\limsup_{\varepsilon\downarrow0} E^Q \Biggl[
\sum_{i=1}^{\infty} 1_{T_{i-1}\leq
T}
Z^{-1}_{T_{i-1}}\lambda_{T_{i-1}}(T_i-T_{i-1})
\frac{\varepsilon^{-\beta
}u^\beta_\varepsilon(\underline
a_{T_{i-1}}, \overline a_{T_{i-1}})}{\varepsilon^{-\alpha}\underline
f_\varepsilon^{0,1}(\underline a_{T_{i-1}}, \overline a_{T_{i-1}})} \Biggr],
\\
S_2 &\geq&\limsup_{\varepsilon\downarrow0} E^Q \Biggl[
\sum_{i=1}^{\infty} 1_{T_{i-1}\leq
T}
Z^{-1}_{T_{i-1}}\lambda_{T_{i-1}}(T_i-T_{i-1})
\frac{\varepsilon^{-\beta
}u^\beta_\varepsilon(\underline
a_{T_{i-1}}, \overline a_{T_{i-1}})}{\varepsilon^{-\alpha}\overline
f_\varepsilon^{0,1}(\underline a_{T_{i-1}}, \overline a_{T_{i-1}})} \Biggr].
\end{eqnarray*}
Using Lemma \ref{weakcvg} we obtain uniform convergence of
\[
\frac{\varepsilon^{-\beta}u^\beta_\varepsilon(a, b)}{\varepsilon
^{-\alpha}\overline
f_\varepsilon^{0,1}(a, b)}
\]
toward $
\frac{u^\beta(a, b)}{
f^{*,0,1}(a, b)}
$ and conclude that
\[
S_2 = E^Q \biggl[\int_0^T
\lambda_t Z_t^{-1} \frac{u^\beta(\underline a_t,
\overline a_t)}{
f^{*,0,1}(\underline a_t, \overline a_t)}\,dt
\biggr]=E \biggl[\int_0^T \lambda_t
\frac{u^\beta(\underline a_t,
\overline a_t)}{
f^{*,0,1}(\underline a_t, \overline a_t)}\,dt \biggr].
\]
\upqed\end{pf}

\begin{appendix}\label{app}
%s7 #&#
\section{Some computations for stable processes}\label{appa}
%
%pr6 #&#
\begin{proposition}\label{squarefunc}
Let $X$ be a symmetric $\alpha$-stable process on $\mathbb R$ with
characteristic
function $E[e^{iuX_t}] = e^{-t\sigma|u|^\alpha}$, {$0<\alpha<2$}, and
$\tau_{a,b} =
\inf\{t\geq0\dvtx X_t\notin(-a,b)\}$ {with $a,b>0$.} Then
\[
f(a,b):= E \biggl[\int_0^{\tau_{a,b}}
X_t^2 \,dt \biggr] = \frac{\alpha(ab)^{1+({\alpha}/{2})}}{2\sigma\Gamma(3+\alpha) } \biggl\{ \biggl(
\frac{a}{b} + \frac{b}{a} \biggr) \biggl(1+\frac{\alpha}{2}
\biggr) - \alpha \biggr\}.
\]
\end{proposition}
The proof of this result is based on the following
lemma, where we consider the exit time from the interval $[-1,1]$
by a process starting from $x$.
%
%le10 #&#
\begin{lemma}
Let $X$ be as above and $\tau_1 =
\inf\{t\geq0\dvtx X_t\notin(-1,1)\}$. Then
\[
f(x):=E^{x} \biggl[\int_0^{\tau_1}
X_t^2 \,dt \biggr] %= \frac{1}{\sigma}
% \Gamma\left(\frac{1}{2}\right)}{\Gamma\left(\frac{4+\alpha}{2}\right)
% \Gamma\left(\frac{3+\alpha}{2}\right) }(1-x^2)^{\frac{\alpha}{2}}
= \frac{1}{\sigma}\frac{2(1-x^2)^{{\alpha}/{2}}
\{x^2 + ({\alpha}/{2}) \}}{\Gamma(3+\alpha)}1_{x\in(-1,1)}.
\]
\end{lemma}
\begin{pf}%{Proof of lemma}
Without loss of generality, we let $\sigma=1$ in this proof.
Let $\hat f(u) = \int_{\mathbb R} e^{iux} f(x) \,dx$. Using the
arguments similar to the ones in \cite{getoor.61}, one can show that
the function $f$ satisfies the equation $\mathcal L^\alpha
f (x)= -x^2$ on $x\in(-1,1)$ with the boundary condition $f(x) = 0$ on
$x\notin(-1,1)$, where $\mathcal L^\alpha$ is the fractional Laplace operator
\begin{eqnarray*}
\mathcal L^\alpha f (x) &=& \int_{\mathbb R}
\bigl(f(x+y)-f(x) - yf'(x)\bigr)\frac{dy}{|y|^{1+\alpha}}, \qquad 1<\alpha<2,
\\
\mathcal L^\alpha f (x) &=& \int_{\mathbb R}
\bigl(f(x+y)-f(x) - y 1_{|y|\leq1}f'(x)\bigr)
\frac{dy}{|y|^{1+\alpha}}, \qquad \alpha=1,
\\
\mathcal L^\alpha f (x) &=& \int_{\mathbb R}
\bigl(f(x+y)-f(x)\bigr)\frac
{dy}{|y|^{1+\alpha}},\qquad  0<\alpha<1.
\end{eqnarray*}
Moreover, the function $\hat
f$ satisfies the system of integral equations
\begin{eqnarray*}
\frac{1}{\pi} \int_0^\infty\hat f(u)
|u|^\alpha\cos(ux) \,du& =& x^2,\qquad |x|<1,
\\
\frac{1}{\pi} \int_0^\infty\hat f(u)
\cos(ux) \,du &=& 0,\qquad |x|>1.
\end{eqnarray*}
Let $\hat f_1 (u) = u^{-{(1+\alpha)}/{2}}
J_{{(1+\alpha)}/{2}}(u)$ and $\hat f_2 (u) = u^{-{(3+\alpha)}/{2}}
J_{{(3+\alpha)}/{2}}(u)$, {where $J$ is the Bessel function; see
\cite{grad}, Section~8.40.} Then, from \cite{grad}, Integral 6.699.2,
we get
%
%e48 #&#
%e49 #&#
%e50 #&#
%e51 #&#
%e52 #&#
%
\begin{eqnarray}
\int_0^\infty\hat f_1(u) \cos(ux)
\,du &=& \int_0^\infty\hat f_2(u)
\cos(ux) \,du = 0,\qquad \mbox{$|x|>1$,}\label{xbig}
\\
\int_0^\infty\hat f_1(u)
|u|^\alpha\cos(ux) \,du& =& 2^{{(\alpha-1)}/{2}} \Gamma \biggl(
\frac{1+\alpha}{2} \biggr),\qquad \mbox {$|x|<1$,} \label{xsmallalpha1}
\\
\label{xsmallalpha2}
\int_0^\infty\hat f_2(u)
|u|^\alpha\cos(ux) \,du &=& 2^{{(\alpha-3)}/{2}} \Gamma \biggl(
\frac{1+\alpha}{2} \biggr) \bigl(1-(1+\alpha)x^2\bigr),
\nonumber
\\[-8pt]
\\[-8pt]
\eqntext{\mbox{$|x|<1$,}}
\\
\label{xsmall1}
\int_0^\infty\hat f_1(u) \cos(ux)
\,du& =&2^{-{(\alpha+1)}/{2}} \frac{\Gamma ({1}/{2} )}{\Gamma
({(\alpha+2)}/{2} )}\bigl(1-x^2
\bigr)^{{\alpha}/{2}},
\nonumber
\\[-8pt]
\\[-8pt]
\eqntext{\mbox {$|x|<1$,}}
\\
\label{xsmall2}
\int_0^\infty\hat f_2(u) \cos(ux)
\,du& =& 2^{-{(\alpha+3)}/{2}} \frac{\Gamma ({1}/{2} )}{\Gamma
({(\alpha+4)}/{2} )}\bigl(1-x^2
\bigr)^{1+({\alpha}/{2})},
\nonumber
\\[-8pt]
\\[-8pt]
 \eqntext{${\mbox{$|x|<1$.}}$}
\end{eqnarray}
From (\ref{xbig})--(\ref{xsmallalpha2}),
\[
\hat f(u) = \pi\frac{\hat f_1(u) - 2\hat
f_2(u)}{2^{{(\alpha-1)}/{2}} \Gamma ({(1+\alpha)}/{2} )
(1+\alpha)}.
\]
To conclude, we compute the inverse Fourier transform of $\hat f$ from
(\ref{xsmall1})--(\ref{xsmall2}).~%
\end{pf}
\begin{pf*}{Proof of Proposition \ref{squarefunc}}
Once again, we set $\sigma=1$
without loss of generality.
Recall a result of Blumenthal, Getoor and Ray
\cite{blumenthal.al.61}: the law of a symmetric stable process
starting from the point $x$ with $|x|<1$ and observed at time
$\tau_1$ has density given by
\[
\mu(x,y) = \frac{1}{\pi} \sin \frac{\pi\alpha}{2} \bigl(1-x^2
\bigr)^{{\alpha}/{2}} \bigl(y^2 -1\bigr)^{-{\alpha}/{2}}
|y-x|^{-1}, \qquad |y|\geq1.
\]
By the scaling property, we then deduce that the density of a symmetric
stable process starting from zero, and observed at time $\tau_{a,b}$
is given by
%
%e53 #&#
%
\begin{equation}
\mu_{a,b}(z) = \frac{1}{\pi} \sin \frac{\pi\alpha}{2}
(ab)^{{\alpha}/{2}} \bigl({(z-b) (z+a)} \bigr)^{-{\alpha}/{2}}
\frac{1}{|z|}. \label{densab}
\end{equation}
Similarly, from the preceding lemma, we easily deduce by the scaling
property that
\[
f_A(x):=E^{x} \biggl[\int_0^{\tau_{A,A}}
X_t^2 \,dt \biggr] = \frac
{2(A^2-x^2)^{{\alpha}/{2}}
\{x^2 + ({\alpha}/{2})A^2 \}}{\Gamma(3+\alpha)}1_{x\in(-A,A)}.
\]
This function satisfies the equation $\mathcal L^\alpha
f_A (x)= -x^2$ on $[-A,A]$ with the boundary condition $f_A(x) = 0$ on
$x\notin[-A,A]$.
Taking $A\geq\mbox{max}(a,b)$, we then get by It\^o's formula
\[
E\bigl[f_A(X_{\tau_{a,b}})\bigr] = f_A(0) - E \biggl[
\int_0^{\tau_{a,b}}X_t^2 \,dt
\biggr].
\]
By symmetry, it is sufficient to prove the proposition for $a\geq
b$. Taking $A=a$ in the above formula, we finally get
\begin{eqnarray*}
&&E \biggl[\int_0^{\tau_{a,b}}X_t^2
\,dt \biggr] \\
&&\qquad= \frac{\alpha
a^{\alpha+2}}{\Gamma(3+\alpha)} - \int_b^a
f_A(x) \mu_{a,b}(x) \,dx
\\
&&\qquad = \frac{\alpha
a^{\alpha+2}}{\Gamma(3+\alpha)} - \frac{2 \sin
{\pi\alpha}/{2}}{\pi\Gamma(3+\alpha)} (ab)^{{\alpha}/ {2}}
\int_b^a \biggl(z^2 +
\frac{\alpha}{2}a^2\biggr) \biggl(\frac{a-z}{
z-b}
\biggr)^{{\alpha}/{2}} \frac{dz}{z}.
\end{eqnarray*}
Computing the integral {(using \cite{grad}, Integral 3.228.1 and the
standard integral representation for the beta function)} then yields
the result.
\end{pf*}

%re8 #&#
\begin{remark}
Let us list here several other useful results which are already known
from the literature or can be obtained with a simple computation.
By a result of Getoor \cite{getoor.61}: under the
assumptions of Proposition \ref{squarefunc},
\begin{eqnarray*}
E^x[\tau_1] &=& \frac{1}{\sigma}\frac{2^{-\alpha}
\Gamma ({1}/{2} )}{\Gamma ({(2+\alpha)}/{2} )
\Gamma ({(1+\alpha)}/{2} ) }
\bigl(1-x^2\bigr)^{{\alpha}/{2}}\\
& =& \frac{1}{\sigma}
\frac{(1-x^2)^{{\alpha}/{2}}}{\Gamma(1+\alpha)}.
\end{eqnarray*}
By the scaling property we then deduce that for general barriers
%
%e54 #&#
%
\begin{equation}
E[\tau_{a,b}] = \biggl(\frac{a+b}{2} \biggr)^\alpha
E^{{(a-b)}/{(a+b)}}[\tau_1] = \frac{(ab)^{{\alpha}/{2}}}{\sigma
\Gamma(1+\alpha)}.\label{tauab}
\end{equation}
Similarly, from \eqref{densab}, we easily get, for $\beta< \alpha$,
%
%e55 #&#
%
\begin{eqnarray}\label{overshootab}
E\bigl[|X_{\tau_{a,b}}|^\beta\bigr] &=& \frac{\sin{\pi\alpha}/{2}}{\pi}(ab)^{{\alpha}/{2}}
\nonumber
\\[-8pt]
\\[-8pt]
\nonumber
&&\times{}\int_0^\infty z^{-\alpha/2}
(z+a+b)^{-\alpha/2} \bigl(|z+a|^{\beta-1} + |z+b|^{\beta-1}
\bigr)\,dz.
\end{eqnarray}
{This integral can be expressed in terms of
special functions and is
equal to
\begin{eqnarray*}
&&a^\beta \biggl(\frac{b}{a+b} \biggr)^{{\alpha}/{2}}
\frac{
\sin{\pi\alpha}/{2}}{\pi}B (1-{\alpha/2},\alpha-\beta ) \\
&&\qquad{}\times F \biggl({\alpha/2}, 1-{
\alpha/2}, {\alpha/2} + 1-\beta, {\frac{b}{a+b}} \biggr)
\\
& &\qquad{}+ b^\beta \biggl(\frac{a}{a+b} \biggr)^{{\alpha}/{2}}
\frac{
\sin{\pi\alpha}/{2}}{\pi}B (1-{\alpha/ 2},\alpha-\beta ) \\
&&\qquad{}\times
F \biggl({\alpha/2}, 1-{\alpha/
2}, {\alpha/2} + 1-\beta, {\frac{b}{a+b}} \biggr),
\end{eqnarray*}
where $B$ is the beta function and $F$ is the hypergeometric
function; see \cite{grad}, Integral 3.259.3.}
\end{remark}
%
%s8 #&#
\section{Convergence of rescaled L\'evy processes}
%
%le11 #&#
\begin{lemma}\label{weakcvg}
Let $X$ be a L\'evy process with characteristic triplet $(0,\nu,
\gamma)$ with respect to the truncation function $h(x) = -1\vee x
\wedge1$ with
\[
x^\alpha\nu((x,\infty)) \to c_+ \quad\mbox{and}\quad x^\alpha
\nu((-\infty,-x)) \to c_- \qquad\mbox{when } x\to0
\]
for some $\alpha\in(1,2)$ and constants $c_+\geq0$ and
$c_-\geq0$ with
$c_+ + c_- >0$. For $\varepsilon>0$, define the process
$X^\varepsilon$ via $X^\varepsilon_t = \varepsilon^{-1}
X_{\varepsilon^\alpha t}$. Then $X^\varepsilon$ converges in law to a strictly
$\alpha$-stable L\'evy process $X^*$ with L\'evy density
%
%e56 #&#
%
\begin{equation}
\nu^*(x) = \frac{c_+ 1_{x>0} + c_- 1_{x<0}}{|x|^{1+\alpha}}.\label{ldens}
\end{equation}
Assume in addition that there exists $C<\infty$, such
that for all $x>0$,
\[
\nu((-x,x)^c) < C x^{-\alpha}
\]
and for $a,b \in(0,\infty)$ and $\beta\in(0,\alpha)$, let
\[
u^\beta_\varepsilon(a,b) = E\bigl[\bigl|X^\varepsilon_{\tau^\varepsilon}\bigr|^\beta
\bigr],\qquad \tau^\varepsilon= \inf\bigl\{t\geq0\dvtx X^\varepsilon_t
\notin(-a,b)\bigr\}.
\]
Then
\[
\lim_{\varepsilon\downarrow0} u^\beta_\varepsilon(a,b) =
u^\beta(a,b)
\]
uniformly on $(a,b) \in[B^{-1},B]^2$ for all $ B < \infty$, with
\[
u^\beta(a,b) = E\bigl[\bigl|X^*_{\tau^*}\bigr|^\beta\bigr]
\]
and $\tau^* = \inf\{t\geq0\dvtx X^*_t \notin(-a, b)
\}$.
\end{lemma}
\begin{pf}
Part (i). From the L\'evy--Khintchine formula it
is easy to see that the characteristic triplet
$(A^\varepsilon,\nu^\varepsilon,\gamma^\varepsilon)$ of
$X^{\varepsilon}$ is given by
\begin{eqnarray*}
A^\varepsilon&=& 0,
\\
\nu^\varepsilon(B) &= &\varepsilon^\alpha\nu\bigl(\{x\dvtx x/
\varepsilon\in B\} \bigr),\qquad B\in\mathcal B(\mathbb R),
\\
\gamma^\varepsilon&=& \varepsilon^{\alpha-1} \biggl\{\gamma+ \int
_{\mathbb R} \nu(dx) \bigl(\varepsilon h(x/\varepsilon)-h(x)\bigr)
\biggr\}.
\end{eqnarray*}
Under the conditions of the lemma, by Theorem VII.2.9 and Remark
VII.2.10 in~\cite{jacodshiryaev}, in order to prove the convergence in
law, we need to check (a) that
\[
\gamma^\varepsilon\to-\frac{c_+-c_-}{\alpha(\alpha-1)},
\]
where the right-hand side is the third component of the characteristic
triplet of the strictly stable process with L\'evy density \eqref{ldens}
with respect to the truncation function $h$, and (b) that $|x|^2
\wedge1 \cdot\nu^\varepsilon(dx)$ converges weakly to $|x|^2
\wedge1 \cdot\nu^*(dx)$.
Since $\alpha>1$ and $h$
is bounded, for $\eta$ sufficiently small, using integration by parts
and the assumption of the lemma, we obtain
\begin{eqnarray*}
\lim_{\varepsilon\downarrow0} \gamma^\varepsilon&=& \lim_{\varepsilon\downarrow0}
\varepsilon^{\alpha-1} \int_{|x|\leq
\eta} \nu(dx) \bigl(
\varepsilon h(x/\varepsilon) - h(x)\bigr)
\\
& =& \lim_{\varepsilon\downarrow0} \varepsilon^{\alpha-1} \biggl\{\int
_{-\eta}^{-\varepsilon} (-\varepsilon-x)\nu(dx) + \int
_{\varepsilon}^\eta(\varepsilon-x)\nu(dx) \biggr\}
\\
& =& \lim_{\varepsilon\downarrow0} \varepsilon^{\alpha-1} \biggl\{\int
_{-\eta}^{-\varepsilon} \nu\bigl([-\eta,x]\bigr)\,dx - \int
_{\varepsilon}^\eta\nu\bigl([x,\eta]\bigr)\,dx \biggr\}
\\
& =& \lim_{\varepsilon\downarrow0} \varepsilon^{\alpha-1} \biggl\{\int
_{-\eta}^{-\varepsilon} \nu\bigl((-\infty,x]\bigr)\,dx - \int
_{\varepsilon}^\eta\nu\bigl([x,\infty)\bigr)\,dx \biggr
\}
\\
& =& \lim_{\varepsilon\downarrow0} \varepsilon^{\alpha-1} \biggl\{\int
_{-\eta}^{-\varepsilon} \frac{c_-}{|x|^\alpha}\,dx - \int
_{\varepsilon}^\eta\frac{c_+}{|x|^\alpha}\,dx \biggr\} = -
\frac
{c_+-c_-}{\alpha(\alpha-1)}.
\end{eqnarray*}
For property (b),
it is sufficient to show that for all $x\geq0$,
\begin{eqnarray*}
\int_x^\infty|z|^2 \wedge1 \cdot
\nu^\varepsilon(dz)& \to& \int_x^\infty|z|^2
\wedge1 \cdot\nu^*(dz)\quad\mbox{and}
\\
\int_{-\infty}^{-x} |z|^2
\wedge1 \cdot\nu^\varepsilon(dz) &\to& \int_{-\infty}^{-x}
|z|^2 \wedge1 \cdot\nu^*(dz).
\end{eqnarray*}
This is done using integration by parts and the assumption of the
lemma as in the previous step.

Part (ii). First, similar to the proof of Proposition
3 in \cite{rosenbaum.tankov.10}, it is easy to show that
$X^\varepsilon_{\tau^\varepsilon}$ converges in law to
$X^*_{\tau^*}$ as $\varepsilon\downarrow0$. To complete the proof
of {the convergence of~$u^\beta_\varepsilon(a,b)$ to $u^\beta(a,b)$ for
fixed $a$ and $b$,} it remains to show that for all $\beta\in
(0,\alpha)$,
\[
E\bigl[\bigl|X^\varepsilon_{\tau^\varepsilon}\bigr|^\beta\bigr]
\]
is bounded uniformly in $\varepsilon$. From Lemma \ref{overshoot.lm},
\[
E\bigl[\bigl|X^\varepsilon_{\tau^\varepsilon}\bigr|^\beta\bigr] \leq C
\varepsilon^{-\alpha} E\bigl[\tau^\varepsilon\bigr]
\]
for some constant $C$ which does not depend on $\varepsilon$. On the
other hand, for $\varepsilon$ small enough,
\[
E\bigl[\tau^\varepsilon\bigr] \leq E\bigl[\inf\bigl\{t\dvtx |\Delta
X_t|\geq \varepsilon(a+b)\bigr\}\bigr] = \frac{1}{\nu((-\varepsilon a,\varepsilon b)^c)} \leq
C' \varepsilon^\alpha
\]
for a different constant $C'$ [the equality above holds because $\inf\{
t\dvtx |\Delta X_t|\geq
\varepsilon(a+b)\}$ is an exponential random variable with parameter
$\nu((-\varepsilon a,\varepsilon b)^c)$ by the L\'evy--It\^o
decomposition].

It remains to show that the convergence is uniform in
$a$ and $b$. {First, let us show that $u^\beta(a,b)$ is continuous
in $(a,b)$ for $(a,b)\in[B^{-1},B]^2$ and therefore also uniformly
continuous on this set. Let $(a_n)$ and $(b_n)$ be two sequences with
$a_n \to a\in [B^{-1},B]$ and $b_n \to b \in[B^{-1},B]$. For any
process $Y$, we write $\tau_{(a,b)}(Y):= \inf\{t\geq0\dvtx Y_t \notin(-a,b)
\}$ and $\mathcal O_{(a,b)}(Y):= Y_{\tau_{(a,b)}(Y)}$. Then
\[
\mathcal O_{(a_n,b_n)}\bigl(X^*\bigr) = \frac{a_n+b_n}{a+b} \mathcal
O_{(a,b)}\bigl(X^n\bigr)\qquad \mbox{where } X^n =
\frac{ba_n -
ab_n}{a_n+b_n} + \frac{a+b}{a_n+b_n} X^*.
\]
Since clearly $X^n$ converges in law (in Skorokhod topology) to $X^*$,
we can once again proceed similar to the proof of Proposition 3 in
\cite{rosenbaum.tankov.10} to show that $\mathcal O_{(a_n,b_n)}(X^*)$
converges in law to $\mathcal O_{(a,b)}(X^*)$. Then, as above, we use the
\mbox{uniform} integrability of $|\mathcal O_{(a_n,b_n)}(X^*)|^\beta$ for
$\beta\in(0,\alpha)$ to show that\break  $E[|\mathcal
O_{(a_n,b_n)}(X^*)|^\beta]$ converges to $E[|\mathcal
O_{(a,b)}(X^*)|^\beta]$. }

Next, letting $\delta>0$, we use the uniform continuity of $u^\beta$ to
choose $\rho$ such that for all $(a,b)$ and $(a',b')$ belonging to
$[B^{-1},B]$, $|a-a'| + |b-b'|\leq\rho$ implies
$|u^\beta(a,b)-u^\beta(a',b')| \leq\delta/2$.

Next, for every $\lambda>0$,
\[
u^\beta_\varepsilon(\lambda a,\lambda b) = \lambda^\beta
u^\beta _{\varepsilon\lambda}(a,b),
\]
which means that $u^\beta_\varepsilon(\lambda a,\lambda b)$ converges
to $u^\beta(\lambda a,\lambda b)$
uniformly on $\lambda\in[\lambda_1,\lambda_2]$ for $0<\lambda_1 <
\lambda_2 <\infty$. For $B^{-1} = a_0 < a_1 < \cdots< a_N = B$ with $a_{i+1}-a_i
\leq\rho$ for $i=0,\ldots, N-1$, this enables us to find $\varepsilon_0$
such that
for all $\varepsilon< \varepsilon_0$, every $i=0,\ldots, N$ and all
$\lambda\in[B^{-2},1]$,
%
%e57 #&#
%
\begin{equation}
\bigl|u^\beta_\varepsilon(\lambda a_i, \lambda B) -
u^\beta(\lambda a_i, \lambda B)\bigr| \leq\frac{\delta}{2}.\label{uniconv}
\end{equation}
Now, let $(a,b)\in[B^{-1},B]$ be arbitrary, but to fix
the ideas, assume without loss of generality that $a\leq b$. Since
$u^\beta_\varepsilon(a,b)$ is increasing in $a$ on $a\leq b$,
\[
u^\beta_\varepsilon(a,b) \in\biggl[u^\beta_\varepsilon
\biggl(a_i\frac
{b}{B},b\biggr),u^\beta_\varepsilon
\biggl(a_{i+1}\frac{b}{B},b\biggr)\biggr],
\]
where $i$ is such that $a_i \leq a\frac{B}{b} \leq a_{i+1}$, and by
the property \eqref{uniconv}, also
\[
u^\beta_\varepsilon(a,b) \in\biggl[u^\beta
\biggl(a_i\frac{b}{B},b\biggr) - \frac{\delta}{2},u^\beta
\biggl(a_{i+1}\frac{b}{B},b\biggr) + \frac{\delta}{2}\biggr].
\]
We finally use the uniform continuity of $u^\beta$ to conclude
that $u^\beta_\varepsilon(a,b)\in
[u^\beta(a,b)-\delta,u^\beta(a,b)+\delta]$.\vadjust{\goodbreak}
\end{pf}

%s9 #&#
\section{A toy model with a continuous component}
Through a toy model, we show in the next proposition that if we include
a continuous local martingale part in $X$,
it dominates the purely discontinuous part.

%pr7 #&#
\begin{proposition}\label{toy}
Assume $(\mathit{HY})$ and there exists $B>0$, $\sigma>0$ and $\alpha'\in(1,2)$
such that $|A_t|\leq B$, $\frac{1}{B}
\leq\underline a_t, \overline a_t \leq B$ and $X$
is a L\'evy process with characteristic triplet $(\sigma^2,\nu,
0)$ with respect to the truncation function $h(x) = -1\vee x
\wedge1$ where $\nu$ is a L\'evy measure with L\'evy density
\[
\nu(x) = \frac{c_+ 1_{x>0} + c_- 1_{x<0}}{|x|^{1+\alpha'}}.
\]
Then Theorems \ref{err.thm} and~\ref{cost.thm} hold with
$\lambda\equiv1$, $\alpha=2$, $\beta<\alpha'$ and $X^*_t=\sigma W_t$,
where $W_t$ is a Brownian motion.
\end{proposition}
\begin{pf}
We first show that Theorem \ref{err.thm} holds with $X^*_t=\sigma W_t$.
We follow the steps of the proof in Section~\ref{proofthm}.
Step 1 follows from the assumptions of the proposition, and there is
now no need to change probability. Also, Lemma \ref{supzero} easily
holds in the setting\vadjust{\goodbreak} of Proposition \ref{toy}. For step 3, note that
$\lambda_t=1$ and therefore
\[
\hat X_t^i=\tilde X_t^i=X_{T_i+t}-X_{T_i},\qquad
\tilde{\tau}_i=T_{i+1}-T_i.
\]
Thus we easily get \eqref{S1tilde} with $Q=P$ and $Z_t=1$.
Then for step 4 we have
\[
E_{\mathcal F_{T_i}} \biggl[ \biggl(\int_0^{\tilde\tau_i} |
\hat X_t|^\kappa \,dt \biggr)^n \biggr]=\underline
f^{\kappa,n}_\varepsilon(\underline a_{T_i},\overline
a_{T_i}),
\]
with $B^2$ taken equal to zero in the definition of $\hat\tau_1$
defining $\underline f^{\kappa,n}_\varepsilon(a,b)$.
Then note from \cite{rosenbaum.tankov.10}, $\tau_1^{\varepsilon}$ has
uniformly bounded polynomial moments of any order and
$X_t^{\varepsilon}$ (with $\alpha=2$) converges toward $\sigma W_t$.
Following the proof of Lemma \ref{flimit.lm}, this gives that
\[
\lim_{\varepsilon\downarrow0}\varepsilon^{-(\kappa+2)}f^{\kappa
,n}_\varepsilon(a,b)=
f^{*,\kappa,n}(a,b).
\]
Finally, we obtain that the preceding convergence is uniform in $(a,b)$
as in steps 4 and~5 follows easily.

In the same spirit, in order to show that Theorem \ref{cost.thm} holds with $X^*_t=\sigma W_t$ and $\alpha=2$, it is enough
to follow the steps of the proof in Section~\ref{proofthm2}. This can
be done as in the preceding paragraph. However, we still need to prove
part (ii) in Lemma \ref{weakcvg} in the case where a Brownian component
is present, meaning we take $X^*_t=\sigma W_t$ for the limiting process
and $\alpha=2$ in the definition of $X^{\varepsilon}_t$. To this end,
remark that in the setting of Proposition~\ref{toy},
\[
\bigl|X^{\varepsilon}_{\tau^{\varepsilon}}\bigr|^\beta\leq c\bigl(1+\bigl|\check
X^{\varepsilon}_{\check\tau^{\varepsilon}}\bigr|^\beta\bigr),
\]
with $\check X_t=X_t-\sigma W_t$ and
$\check\tau^{\varepsilon} = \inf\{t\geq0\dvtx \check X^{\varepsilon}_t
\notin(-(a+b),a+b)\}$. Thus, using Lemma \ref{weakcvg}, we get that
\[
E\bigl[\bigl|X^{\varepsilon}_{\tau^{\varepsilon}}\bigr|^\beta\bigr]
\]
is bounded uniformly in $\varepsilon$. Then we can replicate the end of
the proof of Lemma~\ref{weakcvg}.
\end{pf}
%

%s10 #&#
\section{Proof of Proposition \lowercase{\protect\texorpdfstring{\ref{sde.prop}}{1}}}
\label{proof1}
\begin{pf}
The process $X$ can be written as
\[
X_t = X_0 + \int_0^t
\bar b_s \,ds + \int_0^t \int
_{|z|\leq1} z (M-\mu) (ds \times dz) + \int_0^t
\int_{|z|> 1} z M(ds \times dz),
\]
where $M$ is a random measure whose compensator $\mu$ is given by
$\mu(\omega,dt\times dz) = dt\times\bar\nu(\gamma^{-1}_t(dz)) 1_{z\in
\gamma_t(U)}=
\frac{\bar\nu(\gamma_t^{-1}(z))}{\gamma'_t(\gamma^{-1}(z))} 1_{z\in
\gamma_t(U)} \,dt\times
\,dz$. Hence,
\begin{eqnarray*}
\mu_t((x,\infty))& =& \int_{\gamma^{-1}_t(x)}^\infty
\bar\nu(y) 1_{y\in
U} \,dy, \\
\mu_t((-\infty,-x)) &=& \int
_{-\infty}^{\gamma^{-1}_t(-x)} \bar\nu(y) 1_{y\in
U} \,dy.
\end{eqnarray*}
By assumption \eqref{stable.strong},
\[
\int_{x}^\infty\bar\nu(y) 1_{y\in
U} \,dy =
\frac{c_+}{x^\alpha} + O\bigl(x^{1-\alpha}\bigr)\quad \mbox{and}\quad \int
_{-\infty}^{-x} \bar\nu(y) 1_{y\in
U} \,dy =
\frac{c_-}{x^\alpha} + O\bigl(x^{1-\alpha}\bigr)
\]
as $x\to0$ and
\[
\int_{x}^\infty\bar\nu(y) 1_{y\in
U} \,dy +
\int_{-\infty}^{-x} \bar\nu(y) 1_{y\in
U} \,dy
\leq\frac{C}{x^\alpha}
\]
for some $C<\infty$ and all $x>0$.
On the other hand, by Taylor's theorem, $\gamma^{-1}_t(x) = \frac
{x}{\gamma'_t(x^*)}$
with $x^*\in[0,x]$. Therefore, we easily obtain that for some $C<\infty$,
%
%e58 #&#
%e59 #&#
%
\begin{eqnarray}
x^\alpha\mu_t((x,\infty)) + x^\alpha
\mu_t((-\infty,-x))&\leq& C\max_{x\in U}
\gamma'_t(x)^{\alpha} \qquad\mbox{for all $x$};
\\
\lim_{x\downarrow0} x^\alpha\mu_t((x,
\infty))& =& c_+ \gamma'_t(0)^\alpha\quad
\mbox{and}
\nonumber
\\[-8pt]
\\[-8pt]
\nonumber
 \lim_{x\downarrow0} x^\alpha \mu_t
((-\infty,-x)) &=& c_- \gamma'_t(0)^\alpha,
\end{eqnarray}
which proves assumption $(\mathit{HX})$.

To show $(\mathit{HX}^\rho_{\mathrm{loc}})$, let $\nu$ be a strictly positive L\'
evy density satisfying \eqref{stable.strong},
continuous outside any neighborhood of zero. We need to prove that the
random function $K_t(z)$ defined by
\[
K_t(z) = \frac{\bar\nu(\gamma_t^{-1}(z))1_{z\in\gamma_t(U)}}{\gamma'_t(\gamma
_t^{-1}(z)) \gamma'_t(0)^\alpha\nu(z)},
\]
satisfies the integrability condition \eqref{integrK}. Let $(\tau_n)$
be the
sequence of stopping times from condition \eqref{extracond}, let
$t<\tau_n$ and $\varepsilon$ be small enough so that $\{|z|\leq
\varepsilon\}\subset\gamma_t(U)$, $t\leq\tau_n$. Clearly,
%
%e60 #&#
%
\begin{eqnarray}
\label{5terms} \qquad\int_{\mathbb R} \bigl|\sqrt{K_t(z)}-1\bigr|^{2\rho}
\nu(dz) &\leq& \int_{|z|\leq\varepsilon} \bigl|\sqrt{K_t(z)}-1\bigr|^{2\rho}
\nu(dz)
\nonumber
\\[-8pt]
\\[-8pt]
\nonumber
&&{}+ \int_{|z|>\varepsilon, z\in\gamma_t(U)} K^\rho_t(z) \nu(dz) +
\nu\bigl(\bigl\{z\dvtx |z|> \varepsilon\bigr\}\bigr).
\end{eqnarray}
The third term above is clearly bounded. To deal with the second term,
observe that by the fact that $\nu$ and $\bar\nu$ are continuous
outside any neighborhood of zero, condition \eqref{extracond} and the
fact that $U$ is compact, on the set $\{z\dvtx |z|>\varepsilon,
z\in\gamma_t(U)\}$ for $t\leq\tau_n$,
\[
K_t \leq C^{1+\alpha}_n \frac{\max\{\bar\nu(z)\dvtx z\in U, |z| \geq
\varepsilon/C_n\}}{\min\{\nu(z) \dvtx |z|\geq\varepsilon, z \in
C_n U\} }<\infty.
\]
Therefore, the second term in \eqref{5terms} is also bounded for
$t\leq\tau_n$. We finally focus on the first term in \eqref{5terms}.
First, observe that on the set where $|z|\leq\varepsilon$,
%
%e61 #&#
%
\begin{eqnarray}
\label{3termsK} \bigr|K_t(z)-1\bigr|&\leq& \biggl\llvert \frac{|z|^{1+\alpha} }{|\gamma_t^{-1}(z)|^{1+\alpha} \gamma
'_t(0)^{1+\alpha}} -1
\biggr\rrvert \frac{\gamma'_t(0)}{\gamma'_t(\gamma_t^{-1}(z))
}\frac{|\gamma_t^{-1}(z)|^{1+\alpha}\bar\nu(\gamma_t^{-1}(z))}{
|z|^{1+\alpha}\nu(z)}
\nonumber\\
&&{} +\biggl\llvert \frac{\gamma'_t(0)}{\gamma'_t(\gamma_t^{-1}(z))
}- 1\biggr\rrvert \frac{|\gamma_t^{-1}(z)|^{1+\alpha}\bar\nu(\gamma_t^{-1}(z))}{
|z|^{1+\alpha}\nu(z)} \\
&&{}+\biggl
\llvert \frac{|\gamma_t^{-1}(z)|^{1+\alpha}\bar\nu(\gamma_t^{-1}(z))}{
|z|^{1+\alpha}\nu(z)} - 1\biggr\rrvert.\nonumber
\end{eqnarray}
For the first term in \eqref{3termsK}, by Taylor's formula and using
condition \eqref{extracond},
\begin{eqnarray*}
\biggl\llvert \frac{|z|^{1+\alpha} }{|\gamma_t^{-1}(z)|^{1+\alpha}
\gamma'_t(0)^{1+\alpha}} -1\biggr\rrvert &=& \biggl\llvert
\frac{\gamma_t'(z^*)^{1+\alpha} }{\gamma'_t(0)^{1+\alpha}} -1\biggr\rrvert \leq(1+\alpha) C_n^{2\alpha+1}
\bigl\llvert \gamma_t'\bigl(z^*\bigr) -
\gamma'_t(0)\bigr\rrvert
\\
&\leq&(1+\alpha) C_n^{2\alpha+2} |z|,
\end{eqnarray*}
where $z^*\in[z\wedge0, z \vee0]$. In the second term, similarly,
\[
\biggl\llvert \frac{\gamma'_t(0)}{\gamma'_t(\gamma_t^{-1}(z))
}- 1\biggr\rrvert \leq C_n \bigl|
\gamma'_t\bigl(\gamma_t^{-1}(z)
\bigr) - \gamma'_t(0)\bigr|\leq C_n^2
\bigl|\gamma^{-1}_t(z)\bigr|\leq C^{3n} |z|.
\]
For the third term, it follows from \eqref{stable.strong} that for
some constant $C<\infty$,
\begin{eqnarray*}
\biggl\llvert \frac{|\gamma_t^{-1}(z)|^{1+\alpha}\bar\nu(\gamma_t^{-1}(z))}{
|z|^{1+\alpha}\nu(z)} - 1\biggr\rrvert &\leq&\biggl\llvert
\frac{1+ C |\gamma^{-1}_t(z)|}{1-C |z|} - 1\biggr\rrvert \leq
 \frac{C}{1-C\varepsilon} \bigl(\bigl|
\gamma_t^{-1}(z)\bigr|+|z| \bigr)
\\
&\leq& \frac{C(1+C_n)}{1-C\varepsilon} |z|.
\end{eqnarray*}
In addition, assume that $\varepsilon$ is chosen small enough so that
$C\varepsilon<1$. Therefore,
\[
\bigl|K_t(z)-1\bigr|\leq c_n |z|
\]
for some constant $c_n<\infty$ (which may later change from line to
line). This easily implies that for $\rho\geq1$,
\[
\int_{|z|\leq\varepsilon} \bigl|\sqrt{K_t(z)}-1\bigr|^{2\rho}
\nu(dz) \leq c_n.
\]
\upqed\end{pf}
\end{appendix}

\section*{Acknowledgments}
We are very grateful to the Associate Editor and to the three referees
for their careful reading of the manuscript and their very relevant
remarks.

%
% imsref loaded by akundreckaite, 2013-11-04 15:36:41


%suskaldyti doi

\printaddresses

\end{document}